\documentclass[pra,aps,floatfix,superscriptaddress,11pt,tightenlines,longbibliography,onecolumn,notitlepage,accepted=2018-05-22]{quantumarticle}

\usepackage[utf8]{inputenc}
\usepackage[english]{babel}
\usepackage[T1]{fontenc}
\usepackage{amsmath,amsthm,amsfonts,amssymb,xcolor}
\usepackage{graphicx}
\usepackage{microtype}
\usepackage{braket}
\usepackage[numbers,sort&compress,square]{natbib}

\numberwithin{equation}{section}
\numberwithin{figure}{section}

\newtheorem{theorem}{Theorem}
\newtheorem{lemma}{Lemma}

\newtheorem{definition}{Definition}

\DeclareMathAlphabet{\mathpzc}{OT1}{pzc}{m}{it}

\newcommand{\ebar}{\overline{\epsilon^\mathrm{out}}}

\newcommand{\mF}{{\mathbb{F}}}
\newcommand{\FF}{{\mathbb{F}}}

\newcommand{\mC}{{\mathcal C}}
\newcommand{\mG}{{\mathcal G}}
\newcommand{\mD}{{\mathcal D}}
\newcommand{\spn}{{\rm span}}
\newcommand{\cM}{{\mathcal M}}
\newcommand{\pr}{{\rm Pr}}
\newcommand{\scc}{{\rm success}}
\newcommand{\onev}{\vec {\bf 1}}
\newcommand{\bu}{\vec{\bf u}}
\newcommand{\bv}{\vec{\bf v}}
\newcommand{\bw}{\vec{\bf w}}
\newcommand{\bg}{\vec{\bf g}}
\newcommand{\mI}{{\mathbb I}}
\newcommand{\ein}{\epsilon_{in}}
\newcommand{\eout}{\epsilon_{out}}
\newcommand{\nCCZb}{\overline{n_{CCZ}}}
\newcommand{\nTb}{\overline{n_{T}}}
\newcommand{\nin}{{n_\mathrm{inner}}}
\newcommand{\kin}{{k_\mathrm{inner}}}

\newcommand{\sgate}{S}
\newcommand{\cS}{{\mathcal{S}}}
\newcommand{\nTbar}{{\overline{n_T}}}
\newcommand{\geven}{{G_{0}}}
\newcommand{\mgeven}{{{\mathcal G}_0}}
\newcommand{\mgodd}{{{\mathcal G}_T}}
\newcommand{\godd}{{G_{T}}}
\newcommand{\keven}{{k_0}}
\newcommand{\kodd}{{k_T}}

\begin{document}

\title{Codes and Protocols for Distilling $T$, controlled-$S$, and Toffoli Gates}

\author{Jeongwan Haah}
\affiliation{Quantum Architectures and Computation Group, Microsoft Research, Redmond, WA 98052, USA}

\author{Matthew B.~Hastings}

\affiliation{Station Q, Microsoft Research, Santa Barbara, CA 93106-6105, USA}
\affiliation{Quantum Architectures and Computation Group, Microsoft Research, Redmond, WA 98052, USA}

\begin{abstract}
We present several different codes and protocols 
to distill $T$, controlled-$S$, and Toffoli (or $CCZ$) gates.
One construction is based on codes that generalize the triorthogonal codes of Ref.~\onlinecite{bh}, 
allowing any of these gates to be induced at the logical level by transversal $T$.
We present a randomized construction of generalized triorthogonal codes 
obtaining an asymptotic distillation efficiency $\gamma\rightarrow 1$.
We also present a Reed-Muller based construction of these codes 
which obtains a worse $\gamma$ but performs well at small sizes.
Additionally, we present protocols based on checking the stabilizers of $CCZ$ magic states 
at the logical level by transversal gates applied to codes; 
these protocols generalize the protocols of Ref.~\onlinecite{HHPW}.
Several examples, including a Reed-Muller code for $T$-to-Toffoli distillation, 
punctured Reed-Muller codes for $T$-gate distillation, 
and some of the check based protocols, 
require a lower ratio of input gates to output gates 
than other known protocols at the given order of error correction for the given code size.
In particular, we find a code with parameters $[[512,30,8]]$ that distills $512$ T-gates to $10$ Toffoli gates as well
as triorthogonal codes with parameters $[[887,137,5]],[[912,112,6]],[[937,87,7]]$ with very low prefactors in front of the leading order error terms in those codes.
\end{abstract}

\maketitle

\section{Introduction}
Magic state distillation~\cite{Knill2004a,Knill2004b,BravyiKitaev2005Magic}
is a standard proposed approach to implementing a universal quantum computer.
This approach
begins by implementing the Clifford group
to high accuracy
using either stabilizer 
codes~\cite{Gottesman1996Saturating,CalderbankRainsShorEtAl1997Quantum} 
or using Majorana fermions~\cite{karzig2017scalable}.
Then, to obtain universality, some non-Clifford operation is necessary,
such as the $\pi/4$-rotation (T-gate) or the Toffoli gate (or $CCZ$ which is equivalent to Toffoli up to conjugation by Cliffords).
These non-Clifford operations are implemented
using a resource, called a magic state, which is injected into a circuit that uses only Clifford operations.

Since these magic can produce non-Clifford operations, they cannot themselves be produced by Clifford operations.
Instead, in distillation, the Clifford operations are used to distill a small number of high accuracy magic states
from a larger number of low quality magic state.
There are many proposed protocols to distill magic states:
for $T$ gates from $T$ gates~\cite{Knill2004a,BravyiKitaev2005Magic,MeierEastinKnill2012Magic-state,bh,Jones2012,HHPW}, 
for Toffoli gates from $T$-gates~\cite{eastin2013distilling,jones2013low,Jones2013composite,HHPW,campbell2017prl,campbell2017unified},
for Fourier states from Toffoli gates~\cite{Jones2013fourier},
$CCZ$(Toffoli) states from $CCZ$ gates~\cite{Paetznick2013}.

In such distillation architectures, the resources (space, number of Clifford operations, and number of noisy non-Clifford operations) required to distill magic states far exceed the resources required to implement most quantum algorithms using these magic states.
Hence, improvements in distillation efficiency can greatly impact the total resource cost.

This paper presents a variety of loosely related ideas in distillation.
One common theme is exploring various protocols to distill magic states 
for Toffoli, controlled-$S$, as well as $T$-gates. 
We present several approaches to this.
We use a generalization of triorthogonal codes~\cite{bh} to allow this distillation.
In section \ref{randcon}, we give a randomized construction of such codes 
which achieves distillation efficiency~\cite{bh} $\gamma\rightarrow 1$;
this approach is of some theoretical interest 
because not only is the distance of the code fairly large 
(of order square-root number of qubits) 
but also the least weight stabilizer has comparable weight.
In section \ref{rmcbd}, we give another approach based on Reed-Muller codes.
In addition to theoretical asymptotic results here, 
we also find a particularly striking code 
which distills $512$ $T$-gates into $10$ $CCZ$ magic states 
while obtaining eight order reduction in error.
We also present approaches to distilling Toffoli states 
which are not based on a single triorthogonal (or generalized triorthogonal code) 
but rather on implementing a protocol using a sequence of checks, 
similar to Ref.~\onlinecite{HHPW}.
As in Ref.~\onlinecite{HHPW} we use inner codes to measure various stabilizers of the magic state.
We present two different methods of doing this,
one based on hyperbolic inner codes in section \ref{sec:hyper} 
and one based on normal inner code in section \ref{sec:normal} 
(hyperbolic and normal codes were called even and odd inner codes, respectively, 
in an early version of Ref.~\onlinecite{HHPW}).

In addition to these results for distilling Toffoli states, 
we present other results useful specifically for distilling $T$-gates.
In particular, in \ref{prmc} we study punctured Reed-Muller codes 
and find some protocols with a better ratio of input $T$-gates to output $T$-gates
than any other known protocol for certain orders of error reduction.  Another result in \ref{tradeoff} is a method of reducing the space required for any protocol based on triorthogonal codes at the cost of increased depth.

We use matrices $\sgate = \mathrm{diag} (1, i)$, 
and $T = \mathrm{diag} (1, e^{i \pi/4} )$.
Any subscript $T$ denotes connection to the magic state for $T$ gate.

\section{Triorthogonal Matrices: Definitions and Generalizations}
\subsection{Definitions}
\label{sec:def}

We consider classical codes with $n$ bits, so that code words are vectors in $\mathbb{F}_2^n$.
Given a vector $\bu$, let $|\bu|$ denote the Hamming weight, i.e., the number of nonzero entries of $\bu$.
Given a vector $\bu$, let $\bu_i$ denote the $i$-th entry of $\bu$.
Given two vectors $\bu,\bv$, let $\bu \wedge \bv$ denote the entry wise product of $\bu$ and $\bv$,
i.e., $(\bu \wedge \bv)_i = \bu_i \bv_i$.
Let $\bu \cdot \bv$ denote the inner product, so that $\bu \cdot \bv=\sum_i \bu_i \bv_i$, where the sum is taken modulo $2$.

For us, a classical code $\mC$ will always refer to a linear subspace of $\FF_2^n$.
Given two classical codes $\mC,\mD$, 
let $\mC \wedge \mD$ denote the subspace spanned by vectors $\bu \wedge \bv$ for $\bu \in \mC$ and $\bv \in \mD$.
We will write $\mC^{\wedge 2}$ to mean $\mC \wedge \mC$.
Note that $\mC^{\wedge 2}$ can be a proper superset of $\mC$. 
Given a code $\mC$, let $\mC^\perp$ denote the dual code, i.e. for any vector $\bv$, we have $\bv \in \mC^\perp$ if and only if $\bv \cdot \bu=0$ for all $\bu \in \mC$.
Given two codes, $\mC,\mD$, let $\spn(\mC,\mD)$ denote the span of $\mC$ and $\mD$.

Following Bravyi and Haah~\cite{bh}, 
a binary matrix $G$ of size $m$-by-$n$ is called {\bf triorthogonal}
if 
\begin{align}
\label{c1}
\sum_{j=1}^n G_{a,j} G_{b,j} = 0 \mod 2,
\end{align}
for all pairs $1 \leq a < b \leq m$,
and
\begin{align}
\label{c2}
\sum_{j=1}^n G_{a,j} G_{b,j} G_{c,j} =0 \mod 2,
\end{align}
for all triples of rows $1 \leq a <b <c \leq m$.

Further, we will always assume that the first $\kodd$ rows of $G$ have odd weight, 
i.e. $\sum_{j=1}^n G_{a,j}=1 \mod 2$ for $1 \leq a \leq \kodd$
and the remaining rows have even weight, 
i.e., $\sum_{j=1}^n G_{a,j}=0 \mod 2$ for $\kodd+1 \leq a \leq n$.
(The notation $k_1$ instead of $\kodd$ was used in Ref.~\onlinecite{bh}.)
Let
\begin{align}
\keven=m-\kodd.
\end{align}

Let $\mgeven$ denote the span of the even weight rows of $G$.  
Let $\mgodd$ denote the span of the odd weight rows of $G$.  
Let $\mG$ denote the span of all the rows of $G$.
In association with a triorthogonal matrix
we define a {\bf triorthogonal code}, a quantum CSS code,
by letting $\mgeven$ correspond to $X$-stabilizers,
and $\mG^\perp$ to $Z$-stabilizers.
The distance of a triorthogonal matrix $G$ is defined to be 
the minimum weight of any nontrivial $Z$-logical operators 
of the corresponding triorthogonal code,
i.e., the minimum weight 
of a vector $\bu$ such that $\bu \in \mgeven^\perp$ but $\bu \not \in \mG^\perp$.
The distance of any subspace $\mC$ is defined to be 
the minimum weight of any nonzero vector in that subspace.
Clearly, the distance of a triorthogonal matrix $G$ 
is at least the distance of the subspace $\mgeven^\perp$.

\subsection{Triorthogonal Spaces and Punctured Triorthogonal Matrices}
\label{triopunc}

Let us define a ``triorthogonal subspace'' 
to be a subspace $\mC$ such that for any $\bu,\bv,\bw \in \mC$, we have
$| \bu \wedge \bv \wedge \bw| = 0 \mod 2$.
Given a triorthogonal matrix $G$, the vector space $\mgeven$ is a triorthogonal space.
Thus, any $\keven$-by-$n$ matrix whose rows span $\mgeven$ is a triorthogonal matrix.
However, if $\kodd\neq 0$, then the span of the rows of $G$ is not a triorthogonal space.

In this regard, we note the following.
Let $G$ be an arbitrary triorthogonal matrix of the form
\begin{align}
G=\begin{bmatrix}
\godd \\
\geven
\end{bmatrix},
\end{align}
where $\godd$ is $\kodd$-by-$n$ (and contains the odd weight rows of $G$) and $\geven$ is $\keven$-by-$n$ (and contains the even weight rows of $G$).
Consider the matrix
\begin{align}
\label{tG}
\tilde G=\begin{bmatrix}
I & \godd \\
0 & \geven
\end{bmatrix},
\end{align}
where $I$ denotes a $\kodd$-by-$\kodd$ identity matrix 
and $0$ denotes the zero matrix of size $\keven$-by-$\kodd$.
This matrix $\tilde G$ is a triorthogonal matrix with all rows having even weight,
and its row span defines
a triorthogonal space $\tilde \mG$.
Thus, from a triorthogonal matrix, we can construct a triorthogonal space by adding $\kodd$ additional coordinates to the vector and padding the matrix by $I$.

We now show a converse direction, based on the idea of puncturing a code.
Given any subspace $\tilde \mC$ of dimension $m$,
there exists a matrix $\tilde G$ whose rows form a basis of $\tilde \mC$
(after possibly permuting the coordinates of the space)
such that
\[
\tilde G = \begin{bmatrix} I_m & P \end{bmatrix} = \begin{bmatrix} I_{\kodd} & 0 & P_T \\ 0 & I_{\keven} & P_0 \end{bmatrix},
\]
for some matrix $P$, where $I_m$ is an $m$-by-$m$ identity matrix.
Such a matrix in the {\bf reduced row echelon form} 
is unique once an ordering of coordinate is fixed,
and can be computed by Gauss elimination from any spanning set for $\tilde \mC$.
Choose any $\kodd$ such that $0 \leq \kodd \leq m$.
Let $P_T$ be the first $\kodd$ rows of $P$ and let $P_0$ be the remaining rows of $P$.
Let $\godd = \begin{bmatrix} 0 & P_T\end{bmatrix}$, 
where $0$ is the $\kodd$-by-$\keven$ zero matrix,
and let $\geven=\begin{bmatrix} I_\keven & P_0\end{bmatrix}$,
where $I_\keven$ is the $\keven$-by-$\keven$ identity matrix.
Then, the matrix
\[
\begin{bmatrix} \godd \\ \geven \end{bmatrix} = \begin{bmatrix} 0 & P_T \\ I_\keven & P_0 \end{bmatrix}
\]
is a triorthogonal matrix. We say that this matrix is obtained by ``puncturing" the previous code on the given coordinates.
By the uniqueness of the reduced row echelon form,
the matrices $\godd$ and $\geven$ are determined by 
$\tilde \mC$, $\kodd$, and the ordering of the coordinates.

This idea of padding is related to the following protocol for distillation~\cite{Fowler2012}.
We consider $\kodd = 1$ for the moment, but a generalization to a larger $\kodd$ is straightforward.
Observe that on a Bell pair $\ket \phi = \ket{00} + \ket{11}$ (we ignore global normalization factors),
the action of $T$ on the first qubit is the same as $T$ on the second: $T_1 \ket\phi = T_2 \ket \phi$.
Once we have $T_2 \ket \phi$, suppose we measure out the second qubit onto $\ket +$.
The state on the first qubit is then the magic state $T_1 \ket + = \bra{+_2} T_2 \ket \phi$.
If we instead measure the second qubit in the $\ket -$ state, we can apply a Pauli correction
to bring the first qubit to the desired magic state.
If the second qubit of this Bell pair is a logical qubit of a code, 
where the logical $\bar T$ can be fault-tolerantly implemented, 
then the above observation enables fault-tolerant creation of the magic state.

The protocol is thus as follows.  Consider a triorthogonal code defined by some matrix $G$; for brevity, we also refer to this code as $G$ below.
Let $\tilde \mG$ be the space obtained by padding as above.
(i) Create a Bell pair $\ket{0 \bar 0} + \ket{1 \bar 1}$ where the second qubit is embedded in the code $G$.
The Bell pair is the eigenstate of $X \bar X$ and $Z \bar Z$,
which is simply the state stabilized by $X(v)$ for any $v$ in the triorthogonal space $\tilde \mG$,
and by $Z(v')$ for any $v'$ in $\tilde \mG^\perp$.
Thus, this step can be implemented by a circuit consisting of control-NOTs.
This circuit can be thought of as the preparation circuit of the superposition of all {\it classical} code words of $\tilde \mG$.
(ii) Apply the transversal $T$ gate on the qubits of $G$, followed by possible Clifford corrections;
these Clifford corrections are either phase gate or control-Z~\cite{bh}.
(iii) Project the logical qubit of the code onto a $\ket{ \bar +}$ or $\ket{\bar -}$ state.
This step can be done simply by measuring individual qubits of the code $G$ in the $X$ basis
without inverse-encoding, and classical post-processing.
The reason is that $X$ operator on individual qubits commutes with logical $X$ of the code,
and hence after the $X$ measurement, the state of the qubits that comprised the code 
is some eigenstate of the logical $X$ operator.
The eigenvalue of this logical $X$ can be inferred by taking the parity of the measurement outcomes,
and if necessary we apply a Pauli correction
to the magic state on the other side of the initial Bell pair.
The eigenvalues of the $X$ stabilizers of the code can also be checked similarly and we post-select on these being in the $+$ state.

This protocol is particularly simple to describe in the case that the matrix $G$ is obtained by puncturing some
triorthogonal subspace $\tilde \mG$ on some set of coordinates.  
Then the protocol is: prepare the  superposition of all {\it classical} code words of $\tilde \mG$, 
then apply a transversal $T$ gate on all unpunctured coordinates (followed possibly by a Clifford correction), 
then measure all unpunctured coordinates, and finally, 
if classical postprocessing shows that all stabilizers are in the $+$ state, 
the punctured coordinates are in the desired magic state 
(up to a Pauli correction which is determined by the classical post-processing).

This protocol is different from preparing encoded $\ket{\tilde +}$, applying $\bar T$, and inverse-encoding,
in that the Clifford depth is smaller.
The only Clifford cost is in the initial preparation of the pre-puncture stabilizer state,
and the Clifford correction after $T$.
The Clifford correction after $T$ is absent if the pre-puncture code $\tilde \mG$ is triply even.

\subsection{Generalized Triorthogonal Matrices: $T$-to-$CCZ$ Distillation}
\label{sec:gtrio-def}

Let us now generalize the definition of triorthogonal matrices.
This generalization has appeared in \cite[App.~D]{campbell2017unified},
upon which the ``synthillation" protocols are built.
Our definition is a special case in that we consider only codes that distill $T$-gates, controlled-$S$ gates, and $CCZ$ gates, 
rather than arbitrary diagonal matrices at the third level of the Clifford hierarchy.  
On the other hand, we will present codes of arbitrary distance, 
rather than just distance $2$ of \cite{campbell2017unified}.

\begin{definition}
A $(k_T + 2k_{CS} + 3 k_{CCZ})$-by-$n$ binary matrix $G$ is 
{\bf generalized triorthogonal} if it can be written up to permutations of rows as
\begin{align}
    G = \begin{bmatrix} G_T \\ G_{CS} \\ G_{CCZ} \\ G_0 \end{bmatrix}
    \label{eq:gtrio}
\end{align}
where $G_T$ has $k_T$ rows, $G_{CS}$ has $k_{CS}$ pairs of rows, and $G_{CCZ}$ has $k_{CCZ}$ triples of rows
such that
\begin{align}
    \sum_{i=1}^n G_{a,i} G_{b,i} G_{c,i} \mod 2 = 
    \begin{cases}
        1 & \text{if } a=b=c=1,\ldots,k_T, \\
        1 & \text{if } \begin{cases} 
            a=b=k_T + 2i-1,\\ 
            c= k_T + 2i \end{cases} \text{ for } i = 1,\ldots,k_{CS},\\
        1 & \text{if } \begin{cases} 
            a= k_T + 2k_{CS} + 3i-2,\\
            b = k_T + 2k_{CS} + 3i-1,\\ 
            c= k_T + 2k_{CS} + 3i
            \end{cases} \text{ for } i = 1,\ldots, k_{CCZ},   \\
        0 & \text{otherwise.}
    \end{cases} 
\end{align}
\end{definition}

Such a generalized triorthogonal matrix can be used to distill $n$ $T$-gates 
into $\kodd$ $T$-gates, $k_{CS}$ controlled-$S$ gates, and $k_{CCZ}$ $CCZ$ gates, 
where the $CCZ$ gate is a controlled-controlled-Z gate 
which is conjugate to the Toffoli gate by Clifford operations.
Define a quantum code on $n$ qubits.  Take $X$-type stabilizers of the quantum code which correspond to rows of $\geven$ (i.e., for each row of $\geven$, there is a generator of the stabilizer group which is a product of Pauli $X$ on all qubits for which there is a $1$ entry in that row of $\geven$).
For each row of $\godd,G_{CS}$ and $G_{CCZ}$ there is one logical qubit, with logical $X$-type operators corresponding to the row.  The corresponding $Z$-type logical operators can be determined by the requirement that they commute with the $X$-type stabilizers and by the commutation relations for logical $X$ and $Z$ operators.
Finally, the $Z$-type stabilizers of the code are the maximal set of operators that commutes with all logical operators and $X$-type stabilizers.
It is easy to show, by generalizing the arguments of Ref.~\onlinecite{bh}, that applying a $T$-gate to every qubit will
apply $T$-gates to the logical qubits corresponding to rows of $\kodd$ and will apply controlled-$S$ gates to each pair of logical qubits corresponding to a pair of rows of $G_{CS}$, and will apply $CCZ$ 
gates to each triple of logical qubits corresponding to a triple of rows of $G_{CCZ}$, up to an overall Clifford operation on the logical qubits.
Input errors are detected up to an order given by the distance of the code, where
the distance of a generalized triorthogonal matrix $G$ is defined to be the minimum weight of a vector $\bu$ such that $\bu \in \mgeven^\perp$ and such that $\bu \not \in {\rm span}(\mgodd,{\mathcal G}_{CS},{\mathcal G}_{CCZ})^\perp$, with ${\mathcal G}_{CS},{\mathcal G}_{CCZ}$ being the row spans of $G_{CS},G_{CCZ}$ respectively.

To generalized triorthogonal matrices,
the puncturing and padding in the previous subsection does not immediately carry over.
However, the connection is retained if we consider the puncturing or padding
in the following way.
Let $G$ be a generalized triorthogonal matrix in the form of \eqref{eq:gtrio}
with $k_T$ rows in $G_T$, $k_{CS}$ rows in $G_{CS}$ and $k_{CCZ}$ rows in $G_{CCZ}$,
and let $F$ be another generalized triorthogonal matrix in the form of \eqref{eq:gtrio}
with the same corresponding number of rows in the upper three blocks.
Combine the submatrices of $G$ and $F$ as
\begin{align}
    \tilde G^F :=
\begin{bmatrix}
    F_T & G_T \\
    F_{CS} & G_{CS} \\
    F_{CCZ} & G_{CCZ} \\
    0 & G_0 \\
    F_0 & 0
\end{bmatrix}.
\end{align}
For example, if $G$ is triorthognal where $k_{CS} = k_{CCZ} = 0$,
then $F$ can be such that $F_T = I_{k_T}$ and $F_{CS} = F_{CCZ} = F_{0} = 0$
(this example is precisely the padding in the previous subsection).
Generally, the rows of $\tilde G ^F$ spans a genuine triorthogonal subspace,
and whenever $F$ is canonical (due to e.g. linear independence)
we can recover $G$ from $\tilde G^F$,
the procedure of which amounts to the puncturing.
The distillation protocol of the previous subsection carries over 
to this generalized punctured code;
the only change is that one generally has to inverse-encode 
the logical qubits of the triorthogonal code of $F$.

\subsection{Space-Time Tradeoff For Triorthogonal Codes}
\label{tradeoff}

We now briefly discuss a way of reducing the space required in any protocol based on a triorthogonal code, 
at the cost of increasing circuit depth.
Consider a code with a total of $k$ logical qubits ($k=k_T+2k_{CS}+3k_{CCZ}$), 
a total of $n_X$ $X$-type stabilizer generators, and $n_Z$ $Z$-type stabilizer generators.
The number $n_X$ is equal to the number of rows of $\geven$.
The usual protocol to prepare magic states is to first initialize the logical qubits in the $\ket +$ state,
encode, then apply transversal $T$, measure stabilizers, and, if no error is found, finally decode yielding the desired magic states.
It is possible to implement this protocol using only $k+n_X$ total qubits as follows.

The idea is to always work on the {\it unencoded} state,
but we instead {\it spread} potential errors so that we can detect them.
Recall that encoding is done by preparing a total of $n_X$ ancilla qubits 
in the $\ket +$ state (call these the $X$ ancilla qubits), 
a total of $n_Z$ ancilla qubits in the $\ket 0$ state (call these the $Z$ ancilla qubits), and applying a Clifford.  
Call this Clifford $U$.
Then, an equivalent protocol is: prepare a total of $n_X$ ancilla qubits in the $\ket +$ state,
a total of $n_Z$ ancilla qubits in the $\ket 0$ state, and apply $\prod_{j=1}^n U^\dagger \exp(i\pi Z_j/8) U$, 
then measure whether all the $X$ ancilla qubits are still in the $\ket +$ state.
(There is no need to check the $Z$ ancilla qubits since our error model has only $Z$ errors after twirling.)

The operator $U^\dagger \exp(i \pi Z_j/8) U$ is equal to $\exp(i \pi P_j/8)$ where $P_j = U^\dagger Z_j U$,
which is a product of Pauli $Z$ operators.
Let $P_j = \tilde P_j Q_j$ where $\tilde P_j$ is a product of Pauli $Z$ operators 
on some set of logical qubits (which are {\em not} embedded in a code space!) and $X$ ancilla qubits, 
and $Q_j$ is product of Pauli $Z$ on some set of $Z$ ancilla qubits.
Since the $Z$ ancilla qubits remain in the $\ket 0$ state throughout the protocol,
an equivalent protocol involving only $k+n_X$ total qubits is:
prepare a total of $n_X$ ancilla qubits in the $\ket +$ state,
and apply $\prod_{j=1}^n \exp(i \pi \tilde P_j/8)$, then measure whether all the  $X$ ancilla qubits are still in the $\ket +$ state.
Note that although the product over $j$ ranges from $1$ to $n$, there are only $k+n_X \le n$ physical qubits.

This operator $\exp(i \pi \tilde P_j/8)$ can be applied by a sequence consisting of a Clifford, 
a $T$ gate, and the inverse of the Clifford.
If a subset of $\{ \tilde P_j \}_{j=1}^n$ consists of $n'$ 
(multiplicatively) independent operators, where $n' \le k + n_X$,
then we can apply these $n'$ operators simultaneously by finding a Clifford
that conjugates each of the $n'$ operators to distinct Pauli $Z$ operators.
In the best situation, 
we can obtain a protocol using $k+n_X$ total qubits, 
that requires $\lceil \frac{n}{k+n_X} \rceil$ rounds of Cliffords and $T$-gates.
While the $T$-depth of the circuit is larger than the original protocol, 
the total circuit depth may or may not increase: 
if the Cliffords are implemented by elementary CNOT gates, 
then the circuit depth depends upon the depth required to implement the various encoding and decoding operations.
Other tradeoffs are possible by varying the number of $Z$ ancillas that are kept:
keeping all $Z$ ancillas is the original protocol with minimal depth and maximal space,
while reducing the number will increase depth at the cost of space.

A $Z$ error on a $T$ gate will propagate due to the Cliffords.
Specifically, a Clifford $U^{(j)}$ that maps $\exp(i\pi Z_j / 8)$ to $\exp(i \pi \tilde P_j / 8)$,
will map an error $Z_j$ to $\tilde P_j$, but the error $\tilde P_j$ will not further be affected by
the other $\exp(i \pi \tilde P_j / 8)$ since they commute.
The accumulated error will flip some $X$ ancilla qubits as well as the logical qubits
that would be flipped in the usual protocol.
The association from the errors in $T$ gates to the logical and $X$ ancilla qubits
is identical to the usual protocol.
Hence, in the present space-time tradeoff,
the output error probability and the success probability are
identical to the usual protocol,
whenever the error model is such that only $T$ gates suffer from $Z$ errors.

For example, for the $512$ qubit protocol below to distill $CCZ$ magic states based on $RM(2,9)$,
the number of physical qubits required is $3\times 10+n_X=30+1+9+{9\choose 2}=76$.
For protocols based on a punctured $RM(3,10)$ below, $n_X\leq 1+10+{10 \choose 2}+{10 \choose 3}=176$,
leading in both cases to a large reduction in space required.

\section{Randomized Constructon of Triorthogonal and Generalized Triorthogonal Matrix}
\label{randcon}

We now give a randomized algorithm that either returns a triorthogonal or generalized triorthogonal matrix
with the desired $n,k_{T},k_{CS},k_{CCZ},k_{0}$, or returns failure.
For notational simplicity, we begin with the case of $k_{CS}=k_{CCZ}=0$, i.e., a triorthogonal matrix. 
We then explain at the end how to construct generalized triorthogonal matrices by a straightforward generalization of this algorithm.

\subsection{Randomized Construction of Triorthogonal Matrices}

The matrix is constructed as follows.  
We construct the rows of the matrix iteratively, 
choosing each row uniformly at random subject to constraints given by previous rows.
More precisely, when choosing the $j$-th row of the matrix, 
we choose the row uniformly at random subject to 
(i) the constraint~\eqref{c1} for $b=j$ and for all $a<j$,
(ii) the constraint~\eqref{c2} for $c=j$ and for all $a<b<j$, 
and 
(iii) the constraint that the row has either even or odd weight depending on whether it is one of the first $\kodd$ rows of $G$ or not.
If it is not possible to satisfy all these constraints, 
then we terminate the construction and declare failure.
Otherwise, we continue the algorithm.
If we are able to satisfy the constraints for all rows of $G$,
we return the resulting matrix; in this case, we say that the algorithm ``succeeds.''

Note that all of these constraints that enter into choosing the $j$-th row 
are {\it linear} constraints on the entries of the row. 
Eq.~\eqref{c1} gives $j-1$ constraints 
while Eq.~\eqref{c2} gives $\frac 1 2 (j-1)(j-2)$ constraints (the constraints need not be independent).
We can express these constraints as follows: 
let $\bg_a$ denote the $a$-th row vector of $G$.
Then, let $M_j$ be a $( j-1 +\frac 1 2 (j-1)(j-2) + 1)$-by-$n$ matrix, 
with the first $j-1$ rows of $M_j$ being equal to the first $j-1$ rows of $G$.
The next $\frac 1 2 (j-1)(j-2)$ rows of $M_j$ are vectors $\bg_a \wedge \bg_b$ for $a<b<j$.
The last row of $M_j$ is the all $1$s vector $\onev = [1,1,\ldots,1]$.
Thus, the matrix $M_j$ is determined by $\bg_1,\ldots,\bg_{j-1}$.
The constraints on $\bg_j$ can then be written as
\begin{align}
\label{consodd}
M_j \bg_j = [0,0,\ldots,0,1],
\end{align}
for $1 \leq j \leq \kodd$ and
\begin{align}
\label{conseven}
M_j \bg_j = 0
\end{align}
for $\kodd < j \leq m$.
If $\onev$ is in the span of the first $j-1+\frac 1 2 (j-1)(j-2)$ rows of $M_j$ 
(i.e., all rows but the last of $M_j$),
then the constraints have no solution; otherwise, the constraints have a solution.
Let $\cM_j$ denote the row span of $M_j$; 
then, for $\kodd<j$, the constraint \eqref{conseven} is equivalent to requiring that $\bg_j \in \cM_j^\perp$.

We now analyze the probability that the algorithm succeeds, returning a matrix $G$.
We also analyze the distance of $\mgeven^\perp$.
Our goal is to show a lower bound on the probability that the distance is at least $d$, for some $d$.
The analysis of the distance is based on the first moment method: 
We estimate the probability that a given vector $\bu$ is in $\mgeven^\perp$.
We then sum this probability over all choices of $\bu$ such that $0<|\bu|<d$ and bound the result.

Let $\bu$ be a given vector with $\bu \neq 0$ and $\bu \neq \onev$.
Let us first compute the probability that
$\bu\in \mgeven^\perp$ and $\bu \not \in \cM_m$ conditioned on the algorithm succeeding.
When $\bu \not \in \cM_m$, it holds that $\bu \not \in \cM_j$ for all $j \leq m$.
Hence,
\begin{align}
\Pr[ \bu \cdot \bg_j=0 | \scc \;{\rm and } \;\bu \not \in \cM_j ] = \frac{1}{2},
\end{align}
since $\bu \notin \cM_j$ implies that
the condition $\bu \cdot \bg_j = 0$ is independent of the constraints in \eqref{consodd} or \eqref{conseven}.
Note that success of the algorithm depends only on the choices of the odd weight rows, 
and the even weight rows are chosen after the odd weight rows so that the choice of $\bg_j$ does not affect success.
So,
\begin{align}
\label{probeq}
\Pr[ \bu \in \mgeven^\perp \; {\rm and} \; \bu \not \in \cM_m | \scc]
\leq \prod_{j=\kodd+1}^m \frac{1}{2} = 2^{-\keven}.
\end{align}

Now consider the probability that the algorithm succeeds and $\bu \in \cM_m$. 
As a warm-up, we consider the probability 
that the algorithm succeeds and 
that some vector with small Hamming weight is in $\mG$ (the span of all rows of $G$).
We will use big-O notation from here on, considering the asymptotics of large $n$.
Let $H(x) = -x \log_2 x - (1-x) \log_2(1-x)$ be the binary entropy function.

\begin{lemma}
\label{rowspan}
Consider any fixed $\bu\neq 0$.  Then, the probability that the algorithm succeeds and that $\bu$ is in $\mG$ is
bounded as:
\begin{align}
\label{rowbound}
\Pr\left[ \scc \; {\rm and} \; \bu \in  \mG \right] \leq
\sum_{k=1}^m 2^{-n+k+k (k-1)/2}.
\end{align}
Further, let $\theta \in (0,\sqrt 2)$ be a constant,
and let $c \in (0,\frac 1 2)$ be such that $H(c) = 1 - \frac 1 2 \theta^2$. 
Then, for $m \leq \theta \sqrt n$,
we have
\begin{align}
\label{mosn}
\Pr \left[
\scc \; {\rm and} \; 
\exists \bv  \in \mG \setminus \{0\} \; {\rm s.t.} 
\Bigl( |\bv| \leq cn - o(n) \; {\rm or} \; |\bv| \geq (1-c) n + o(n) \Bigr) 
\right] = o(1).
\end{align}
(This equation means that for some function $f(n)$ which is $o(n)$ 
the probability that there exists a nonzero $\bv \in \mG$ with $|\bv| \leq cn-f(n)$ or $|\bv|\geq (1-c)n+f(n)$ is $o(1)$).
\end{lemma}
\begin{proof}
Suppose $\bu$ is in $\mG$.  
Then, $\bu=\sum_{i=1}^m b_i \bg_i$ for some coefficients $b_i \in \FF_2$.
We consider each of the $2^m-1$ possible nonzero choices of the vector $b$
and bound the probability that, for the given choice of $b$,
$\bu=\sum_{i=1}^m b_i \bg_i$ for $\bg$ chosen by the algorithm.
For a given choice of nonzero $b$, 
let $k$ be the largest $i$ such that $b_i \neq 0$.  
The vector $\bg_k$ is chosen randomly subject to $\frac 1 2 k (k-1) + 1$ constraints.
Hence, for given $\bg_1,\ldots,\bg_{k-1}$ and given $b$, $\bu$, the probability that $\bg_k=\bu+\sum_{i=1}^{m-1} b_i \bg_i$ is bounded by $2^{-n+k(k-1)/2+1}$.
There are $2^{k-1}$ possible choices of $b_1,\ldots,b_{k-1}$.  Summing over these choices and summing over $k$,
Eq.~(\ref{rowbound}) follows.

By a first moment bound, the probability that there is a nonzero vector of weight at most $w$ in $\mG$ is bounded by
\[
\Bigl( \sum_{j=1}^w {n \choose j}\Bigr) \Bigl( \sum_{k=1}^m 2^{-n+k+k(k-1)/2}\Bigr).
\]
Similarly, the probability that there is a vector with weight at least $n-w$ in $\mG$ is 
bounded by
\[
\Bigl( \sum_{j=0}^w {n \choose j}\Bigr) \Bigl( \sum_{k=1}^m 2^{-n+k+k(k-1)/2}\Bigr).
\]
For $m \le \theta \sqrt{n}$ with $\theta<\sqrt{2}$, the exponent reads $-n + m +m (m-1)/2 \le -(1-\theta^2/2)n+o(n)$.
The number of vectors with weight at most $cn$ or at least $(1-c)n$, is $2^{H(c)n+o(n)}$.
By choosing $w = cn$ such that $H(c)-(1-\frac 1 2 \theta^2)=0$, the first moment bound gives a result which is $2^{o(n)}$.
We can instead choose $w = cn - f(n) = ( c - o(1))n$, where $f(n)=o(n)$ is some positive function,
so that $H(c - o(1)) - (1 -\frac 1 2 \theta^2) < 0$
and the first moment bound gives a result which is $o(1)$.
\end{proof}

\begin{lemma}
\label{wedgebound}
Let $\theta$ and $c$ be chosen as in Lemma~\ref{rowspan},
and suppose $m \leq  \theta \sqrt{n}$.
Let $0 < \rho < \frac 1 2$ be a constant.
Then, the probability that the algorithm succeeds and that the (classical) minimum distance
of the subspace $\cM_m$ is smaller than $\rho n$ is at most
\[
2^{H(\rho)n-cn+\theta^2 n+o(n)}  + o(1).
\]
For sufficiently small $\theta > 0$, 
there are $\rho, c > 0$ such that this expression tends to zero for large $n$.
\end{lemma}
\begin{proof}
We say that $\mG$ has {\it good distance} if all nonzero vectors $\bu$ in $\mG$ have $cn - o(n) \leq |\bu| \leq (1-c)n + o(n)$,
where $o(n)$ term is from Lemma~\ref{rowspan}.
By Eq.~\eqref{mosn}, the probability that the algorithm succeeds and that $\mG$ does {\it not} have good distance is $o(1)$.

Let $\bu\neq 0,\onev$.
We now bound the probability that the algorithm succeeds and that $\mG$ has good distance and that $\bu \in \cM_m$.

If $\bu \in \cM_m$, then
for some $m$-by-$m$ binary upper triangular matrix $A_{ij}$ and for some $a \in \FF_2$, we have
\begin{align}
\label{uis}
\bu=\sum_{i,j \;{\rm s.t.} \; i \leq j} A_{ij} \bg_i \wedge \bg_j + a \onev.
\end{align}
We consider each of the $2^{m(m-1)/2}-1$ possible nonzero choices of the matrix $A$ and each of the two choices of $a$,
and bound the probability that Eq.~(\ref{uis}) holds for the given choice.

Suppose $a=0$ (the case $a=1$ follows from this case by considering the vector $\bu+\onev$).
For a given choice of $A$, let $k$ be the largest such that $A_{ik} \neq 0$ for some $i \leq k$.
Let $\bg_1,\ldots,\bg_{k-1}$ be given; we compute the probability that $\bg_k$ is such that Eq.~(\ref{uis}) holds.
Eq.~\eqref{uis} imposes an inhomogeneous linear constraint on $\bg_k$ as
\begin{align}
    \label{vwidentity}
    \bv = \bw \wedge \bg_k
\end{align}
where
\begin{align*}
\bv &= \bu + \sum_{i,j\; {\rm s.t.} i \leq j < k} A_{ij} \bg_i \wedge \bg_j,\\
\bw &= \sum_{i \; {\rm s.t.} \; i < k} A_{ik} \bg_i + A_{kk} \onev.
\end{align*}
Assuming $\mG$ has good distance, we have $ |\bw| \ge cn - o(n) $.
Then, the linear contraint Eq.~\eqref{vwidentity} has rank at least $cn - o(n)$; 
in fact, it fixes at least $cn-o(n)$ components of $\bg_k$.
The vector $\bg_k$ is chosen randomly subject to $\frac 1 2 k (k-1) + 1$ linear constraints.
Hence, the probability that Eq.~(\ref{vwidentity}) holds is at most
\[
2^{-cn + k(k-1)/2 + o(n)}.
\]
Summing over all choices of $A_{ij}$, the probability 
that the algorithm succeeds and that $\mG$ has good distance and that $\bu \in \cM_m$ is bounded by
\[
2^{-cn + m(m-1) + o(n)}.
\]
The number of vectors $\bu$ with $|\bu| \leq \rho n$ is (for $\rho\leq \frac 1 2$)
\begin{align}
\sum_{1 \leq j \leq \rho n} {n \choose j} = 2^{H(\rho) n+o(n)}.
\end{align}
Hence, by a first moment argument, the probability that the algorithm succeeds and that $\mG$ has good distance
and that $\cM_m$ has distance smaller than $\rho n$ for $\rho\leq 1/2$ is
\[
2^{H(\rho)n-cn+m(m-1)+o(n)}.
\]
We know $c \to \frac 1 2$ as $\theta \to 0$.
For small enough $\theta$ we have $-cn+m(m-1)=-\Omega(n)$.
Hence this probability is $o(1)$ for sufficiently small $\rho$.
\end{proof}

Finally,
\begin{lemma}
\label{succprob}
Let $m \leq  \theta \sqrt{n}$.  
Then, for sufficiently small $\theta > 0$, the algorithm succeeds with probability $1-o(1)$.
\end{lemma}
\begin{proof}
Suppose the algorithm fails on step $k \le k_T$.
(The algorithm never declares failure after $k_T$ steps 
as the constraints become a homogeneous linear equation.)
Then, the first $k-1$ steps of the algorithm succeed and
the vector $\onev$ must be in the span of $\{ \bg_i \wedge \bg_j : i \le j < k \}$.
The probability that this happens is $o(1)$, 
as we can see using the same proof as in Lemma~\ref{wedgebound}.
There is one minor modification to the proof: Eq.~\eqref{uis} should be replaced by
\begin{align}
\label{uis2}
\onev=\sum_{i,j \;{\rm s.t.} \; i \leq j < k} A_{ij} \bg_i \wedge \bg_j.
\end{align} 
Also, there is no need to sum over vectors $\bu$ as instead we are considering the probability that a fixed vector is in the span of $\{ \bg_i \wedge \bg_j : i \le j < k \}$.
Otherwise, the proof is the same.
\end{proof}

Hence,
\begin{theorem}
We can choose $m=\Theta(\sqrt{n})$ and $\keven=\Theta(\sqrt{n})$, 
so that $\kodd=m-\keven = \Theta(\sqrt n)$
and with high probability the algorithm succeeds and the triorthogonal matrix $G$ has distance 
\[
d = \Omega\Bigl(\frac{\sqrt{n}}{\log n}\Bigr).
\]
\end{theorem}
\begin{proof}
By Lemma~\ref{succprob}, the algorithm succeeds with high probability for sufficiently small $\theta>0$.
By Lemma~\ref{wedgebound}, for sufficiently small $\theta>0$, for $m=\lfloor \theta \sqrt n \rfloor$,
the distance of $\cM_m$ is $\Theta(n)$ with high probability. 
Now we condition on the event that the algorithm succeeds and $\cM_m$ has linear distance.

The distance of the triorthogonal matrix $G$ can be bounded by a first moment bound.
Since $\cM_m$ has linear distance,
the event that $\bu \in \cM_m$ for any nonzero $\bu$ of weight $o(n)$ does not happen.
Then, we can apply Eq.~\eqref{probeq} using the fact that
for any constant $C$, the number of vectors with weight at most 
$C {\sqrt{n}}/{\log{n}}$ is $2^{C\sqrt{n}/2 +o(1)}$. 
So, for sufficiently small $C$ the first moment bound implies that 
the probability that there is $\bu \in \mG_0^\perp$ of weight $\le C \sqrt n / \log n$ is $o(1)$.
\end{proof}

Now that in this regime, the distillation efficiency~\cite{bh} 
defined as $\gamma=\log(n/\kodd)/\log(d)$ converges to $1$ as $n \to \infty$.

\subsection{Randomized Construction of Generalized Triorthogonal Matrices}
The randomized construction of triorthogonal matrices above 
immediately generalizes to a randomized construction of generalized triorthogonal matrices.  
In the previous randomized construction, 
each vector $\bg_j$ was chosen at random subject to certain linear constraints.
Note that Eqs.~(\ref{conseven},\ref{consodd}) have the same left-hand side but different right-hand side.
These constraints were homogeneous for row vectors in $\geven$ 
(i.e., Eq.~(\ref{conseven}) has the zero vector on the right-hand side) 
and inhomogeneous for row vectors in $\godd$ 
(i.e., Eq.~(\ref{consodd}) has one nonzero entry on the right-hand side).
For a generalized triorthogonal matrix, we follow the same randomized algorithm as before 
except that we modify the constraints on the vectors $\bg_j$.
The vectors will still be subject to linear constraints  that $M_j \bg_j$ is equal to some fixed vector, with the same $M_j$ as before.  
However, the fixed vector is changed in the generalized algorithm to obey the definition of a generalized triorthogonal matrix.
This modifies the success probability of the algorithm, 
but one may verify that the algorithm continues to succeed with high probability in the regime considered before.

\section{Reed-Muller Code Based Distillation}
\label{rmcbd}

\subsection{Review of classical Reed-Muller codes}

The space of $\FF_2$-valued functions over $m$ binary variables $x_1, \ldots, x_m$
is a vector space of dimension $2^m$,
and every such function can be identified with a polynomial in $x_1,\ldots,x_m$.
We choose a bijection $\{ f: \FF_2^m \to \FF_2 \} = \FF_2^{2^m}$
defined by
\begin{align}
   \text{function } f : \FF_2^m \to \FF_2 \Longleftrightarrow \text{codeword }( f(z) )_{z \in \FF_2^m }
\end{align}
where the right-hand side is the list of function values.
In this bijection, the ordering of elements of $\FF_2^m$ is implicit,
but a different ordering is nothing but a different ordering of bits, 
and hence as a block-code it is immaterial. 
For example, the degree zero polynomial $f(x_1,\ldots,x_m) = 1$ is a constant function,
that corresponds to all-1 vector of length $2^m$,
and a degree 1 polynomial $f(x_1,\ldots,x_m) = x_1$ is a function
that corresponds to a vector of length $2^m$ and weight $2^{m-1}$.
Since the variables $x_i$ are binary, we have $x_i^2 = x_i$,
and every polynomial function is a unique sum of monomials 
where each variable has exponent 0 or 1.

For an integer $r \ge 0$
the Reed-Muller code $RM(r,m) \subseteq \FF_2^{2^m}$
is defined to be the set of all polynomials (modulo the ideal $(x_1^2-x_1,x_2^2-x_2,\ldots)$)
of degree at most $r$,
expressed as the lists of function values.
\begin{align}
    RM(r,m) = \{ ( f(x) )_{x \in \FF_2^m } ~|~ f \in \FF_2[x_1,\ldots,x_m] / (x_i^2-x_i), \quad \deg f \le r \}
\end{align}
By definition, $RM(r,m) \subseteq RM(r+1,m)$.
For example, $RM(0,m)$ is the repetition code of length $2^m$.
A basis of $RM(r,m)$ consists of monomials
that are products of at most $r$ distinct variables.
Hence, the number of encoded (classical) bits in $RM(r,m)$ 
is equal to $\sum_{j=0}^r \binom{m}{j}$.
The code distance of $RM(r,m)$ is $2^{m-r}$, 
which can be proved by induction in $m$.

A property we make routine use of is that
whenever a polynomial does not contain $x_1 \cdots x_m$ (the product of all variables),
the corresponding vector of length $2^m$ has even weight.
This allows us to see that the dual of $RM(r,m)$ is again
a Reed-Muller code, and direct dimension counting shows that 
\begin{align}
RM(r,m)^\perp = RM(m-r-1,m).
\end{align}
For Reed-Muller codes it is easy to consider the {\it wedge} product of two codes,
which appears naturally in the triorthogonality.
Namely, given two binary subspaces $V$ and $W$,
we define the wedge product as
\begin{align}
    (v \wedge w)_i &= v_i w_i \text{ where } v,w \in \FF_2^n,\\
    V \wedge W &= \spn_{\FF_2} \{ v \wedge w : v \in V, w \in W \}. \label{eq:wedge-def}
\end{align}
By definition, $V^{\wedge 2} := V \wedge V \supseteq V$.
Since a code word of a Reed-Muller code is a list of function values,
we see that
\begin{align}
    RM(r,m)\wedge RM(r',m) = RM(r+r',m).
\end{align}
It follows that $RM(r,m)$ is triorthogonal subspace if $3r < m$.
(In fact, it is triply even.)

Since a basis of Reed-Muller codes consists of monomials
where each variable has exponent 0 or 1,
it is often convenient to think of a monomial as a binary $m$-tuple,
that specifies which variable is a factor of the monomial.
For example, if $m=3$,
the constant function $f=1$ can be represented as $(0,0,0)$,
the function $f=x_1$ can be represented as $(1,0,0)$,
and the function $f=x_2x_3$ can be represented as $(0,1,1)$.
This $m$-tuple is called an {\bf indicator vector}.
(In contrast to what the name suggests,
the ``sum'' of indicator vectors is not defined.)
An indicator vector $a$ that defines a monomial
corresponds to a code word $\mI_a \in \FF_2^{2^m}$.
Under the wedge product of two code words,
the corresponding two monomials is multiplied.
In terms of indicator vector, 
this amounts to taking bit-wise OR operation
which we denote by $\vee$:
\begin{align}
    \mI_a \wedge \mI_b = \mI_{a \vee b}.
\end{align}
For example, if $m=3$,
\begin{align*}
    a &= (1,0,1) &\leftrightarrow f &= x_1x_3& \leftrightarrow \mI_a &= [00000101]\\
    b &= (1,1,0) &\leftrightarrow f &= x_1x_2& \leftrightarrow \mI_b &= [00000011]\\
    a\vee b &= (1,1,1)  &\leftrightarrow f &= x_1 x_2 x_3& \leftrightarrow \mI_{a \vee b} &= [00000001]
\end{align*}

\subsection{Triorthogonal codes for CCZ}

In \cite{eastin2013distilling,jones2013low},
a construction was presented to distill a single Toffoli gate from $8$ $T$ gates,
so that any single error in the $T$ gates is detected.
More quantitatively, if the input $T$ gates have error probability $\ein$,
the output Toffoli has error probability $\eout = 28 \ein^2 + O(\ein^3)$.
A protocol of similar performance based on a generalized triorthogonal matrix
was presented in \cite{campbell2017unified}.
In this subsection, we present alternatives to these constructions
that builds upon Reed-Muller codes, yielding higher order error suppression.
The protocol of \cite{campbell2017unified} will be the same as our smallest instance.

Let $m$ be a multiple of 3.
We consider $RM(r = m/3 -1,m)$ to 
build a generalized triorthogonal code on $2^m$ qubits,
with $\kodd = k_{CS} = 0$ but $k_{CCZ} > 0$.
Since $3r = m-3 < m$, the generating matrix of $RM(m/3-1,m)$ qualifies to be $\geven$.
The $Z$-distance of the triorthogonal code is 
at least the distance of $RM(m/3-1,m)^\perp = RM(2m/3,m)$,
which is $2^{m/3}$. (In fact, it is exactly this for the following constructions.)

We choose triples of $G_{CCZ}$ specified by triples of indicator vectors $a^{(i)},b^{(i)},c^{(i)}$.
The triorthogonality conditions can be summarized as follows.
\begin{align}
    |a^{(i)}| & \le m/3 + 1, \nonumber\\
    |b^{(i)}| & \le m/3 + 1, \nonumber\\
    |c^{(i)}| & \le m/3 + 1, \nonumber\\
    |a^{(i)} \vee b^{(j)}| &\le 2m/3, \\
    |b^{(i)} \vee c^{(j)}| &\le 2m/3, \nonumber\\
    |c^{(i)} \vee a^{(j)}| &\le 2m/3, \nonumber\\
    |a^{(i)} \vee b^{(j)} \vee c^{(\ell)} | & \begin{cases} = m &\text{if } i = j = \ell, \\ < m &\text{otherwise.} \end{cases} \nonumber
\end{align}
(A similar set of conditions for $G_{CS}$ should be straightforward.)
We choose $a^{(i)},b^{(i)},c^{(i)}$ to have weight exactly $m/3$,
so that the first six conditions above are automatically satisfied.

We will give three constructions of triples obeying these requirements.  
One construction will be analytic, one will be numerical, 
and one will be a randomized construction using the Lovasz local lemma.
It may be useful for the reader to think of a vector $a_i \in \mF_2^m$ 
as corresponding to a subset $A_i$ of some set $S$ with $|S|=m$.  
Then, a triple consists of three disjoint subsets $A_i,B_i,C_i$ of cardinality $m/3$ each.

The analytic construction is as follows:
\begin{align}
    a^{(u)} = (u, \bar u, 0),\quad b^{(u)} = (0, u, \bar u) ,\quad c^{(u)} = (\bar u, 0, u)
\end{align}
where we labeled a triple by $u \in \FF_2^{m/3} \setminus \{ 0, \onev \}$.
So, we have $2^{m/3}-2$ triples.
Here, $(u,\bar u,0)$ denotes the indicator vector of length $m$ 
formed by concatenating three bit strings of length $m/3$,
and $\bar u$ is the complement of $u$ so that $\bar u_i=1-u_i$.
By construction, one can verify that 
$a^{(u)} \vee b^{(u)} \vee c^{(u)} =\onev$ for any $u \in \FF_2^{m/3}$.
The case $u=0$ and $u=\onev$ are excluded, 
for the triple to satisfy the other generalized triorthogonality conditions.
Suppose that $x,y,z$ are rows of $G_{CCZ}$ and are not all from the same triple.
We need to check that $|x\vee y \vee z| < m$.
This condition is violated only if
$x=(u_x,\bar u_x,0)$ and $y=(0,u_y, \bar u_y)$ and $z=(\bar u_z,0,u_z)$ 
for some $u_x,u_y,u_z$ because there is no way to have 
$\onev = 0 \vee 0 \vee u \in \FF_2^{m/3}$
unless $u = \onev$, which case we have excluded.
But then, we must have that $u_x=u_y=u_z$ to have $|x \vee y \vee z | = m$.

In the particular case $m=3$, this construction gives $k_{CCZ}=0$.
However, we can instead have $k_{CCZ}=1$ with the triple of indicator vectors $(1,0,0),(0,1,0),(0,0,1)$,
corresponding to polynomials $x_1, x_2, x_3$.
The full generalized triorthogonal matrix is
\begin{align}
\underbrace{
\begin{pmatrix}
x_1\\
x_2\\
x_3\\
\hline
1
\end{pmatrix}
}_\text{polynomial notation} = 
\underbrace{    
\begin{bmatrix}
 0 & 0 & 0 & 0 & 1 & 1 & 1 & 1 \\
 0 & 0 & 1 & 1 & 0 & 0 & 1 & 1 \\
 0 & 1 & 0 & 1 & 0 & 1 & 0 & 1 \\
 \hline
 1 & 1 & 1 & 1 & 1 & 1 & 1 & 1 
    \end{bmatrix}
}_\text{binary notation}
\end{align} 
where the part above the line is $G_{CCZ}$ and that below the line is $G_0$.
This triorthogonal matrix is maximal in the sense that $(G^{\wedge 2})^\perp = G_0$.
The resulting distillation routine has error probability $28p^2 + O(p^3)$
if input $T$-states have error probability $p$.
This protocol was given in \cite{campbell2017unified},
and is similar to those of \cite{eastin2013distilling,jones2013low}

For $m=6$, we find $n=64,k_{CCZ}=2$ and distance $4$,
of which the triples in terms of polynomials are 
$\{ x_1 x_2, x_3x_4, x_5 x_6\}$ and $\{x_2x_3, x_4x_5, x_6 x_1\}$.
We examined $m=6$ instance further to see 
if there could be more logical qubits extending the two triples,
but found that there does not exist any extra solution to the generalized triorthogonality equations.
Instead, we were able to extend $G_0$.
The resulting generalized triorthogonal matrix, denoting each row by a polynomial, is
\begin{align}
    \begin{pmatrix}
        x_1 x_2, ~ x_3 x_4, ~ x_5 x_6 \\
        x_2 x_3, ~ x_4 x_5, ~ x_6 x_1 \\
    \hline
        1, ~x_1,  ~x_2,  ~x_3, ~ x_4, ~ x_5,~  x_6 \\
x_1 x_3, ~x_3 x_5,~ x_1 x_5,~ x_1 x_3 x_5,  
\end{pmatrix}.
\end{align}
This triorthogonal matrix is also maximal in the sense that $(G^{\wedge 2})^\perp = G_0$.
The leading term in the output error probability is $2944 p^4$.
The coefficient was obtained by brute-force weight enumeration and MacWilliams identity.
This protocol is similar to that of \cite{Jones2013composite},
but not identical;
the 64$T$-to-2$CCZ$ protocol here has a smaller coefficient 
in the output error probability than that of \cite{Jones2013composite}.
If the efficiency measure of a distillation protocol is the ratio of the number of input $T$ gates
to the number of output $CCZ$ gates at a given order of error reduction,
then composing a quadratic $T$-to-$T$ protocol such as those of \cite{bh} and the 8$T$-to-1$CCZ$ protocol above 
is better than the 64$T$-to-2$CCZ$ protocol here.

For $m=9$ we find $n=512$, $k_{CCZ}=6$ and distance $8$.
We then did a numerical search to see if it would be possible to have a larger $k_{CCZ}$, 
restricting to the case that 
the triples of $G_{CCZ}$ are associated with triples of indicator vectors of weight $m/3$.
We were able to find $k_{CCZ}=10$,
and further extend $G_0$ to make the resulting triorthogonal matrix maximal
in the sense that $(G^{\wedge 2})^\perp = G_0$.
The resulting $[[512,30,8]]$ code is the following.
\begin{align}
G_{CCZ} &= \begin{pmatrix}
 x_4 x_5 x_7, & x_2 x_6 x_8, & x_1 x_3 x_9 \\
 x_4 x_5 x_9, & x_2 x_7 x_8, & x_1 x_3 x_6 \\
 x_3 x_4 x_6, & x_1 x_5 x_8, & x_2 x_7 x_9 \\
 x_1 x_8 x_9, & x_3 x_4 x_7, & x_2 x_5 x_6 \\
 x_2 x_5 x_9, & x_1 x_3 x_4, & x_6 x_7 x_8 \\
 x_1 x_4 x_5, & x_2 x_3 x_8, & x_6 x_7 x_9 \\
 x_3 x_5 x_6, & x_1 x_2 x_7, & x_4 x_8 x_9 \\
 x_1 x_3 x_8, & x_2 x_4 x_9, & x_5 x_6 x_7 \\
 x_2 x_3 x_5, & x_1 x_7 x_9, & x_4 x_6 x_8 \\
 x_3 x_8 x_9, & x_1 x_5 x_7, & x_2 x_4 x_6 
\end{pmatrix}\\
G_0 &= \begin{pmatrix}
 RM(r=2,m=9) \\
 x_1 x_2 x_9 \\
 x_1 x_2 x_8 \\
 x_6 x_8 x_9 \\
 x_3 x_7 x_8  
    \end{pmatrix}
\end{align}
Here, each line in $G_{CCZ}$ contains a triple of polynomials 
(actually monomials).
The algorithm we used was as follows.
We used a version of the algorithm in the constructive proof of the Lovasz local lemma
of Ref.~\onlinecite{mosertardos}.
We define a subroutine to initialize a triple, which, for given $i$, 
sets $a^{(i)},b^{(i)},c^{(i)}$ to be random indicator vectors of weight $m/3$ each,
subject to the constraint that $a^{(i)} \vee b^{(i)} \vee c^{(i)} = \onev$. 
That is, ``initializing a triple'' is to choose $a^{(i)}$ at random of weight $m/3$, 
and then choose $b^{(i)}$ at random of weight $m/3$ with its $1$ entries only in the $0$ entries of $a^{(i)}$, 
after which $c^{(i)}$ is determined by the constraint $a^{(i)} \vee b^{(i)} \vee c^{(i)} = \onev$.
Then we do the following:
\begin{itemize}
\item[{\bf 1.}]
Pick $k_{CCZ}$ and initialize $k_{CCZ}$ different triples.

\item[{\bf 2.}] 
Look for a violation of the triorthogonality conditions. 
We check rows $x,y,z$ of the matrix in lexicographic order.
A violation is when $x,y,z$ are not in the same triple but $x \vee y \vee z = \onev$.

\item[{\bf 3.}]
If a violation of the conditions exists for vectors $x,y,z$,
then we find the triples containing $x,y,z$, and initialize those (at most three) triples,
and go to {\bf 2}.
If no violation exists, exit the algorithm, reporting success.
\end{itemize}
We run this algorithm until it reports success or until we give up and terminate the algorithm.
We also tried a slight modification of the algorithm,
in which we did some random permutation of the triples at various steps
(this has an effect similar to randomizing the order in which we check the conditions).

\subsection{Lovasz Local Lemma}
The randomized numerics above used an algorithm in the constructive proof of the Lovasz local lemma.
Here, we show what the local lemma implies about the possible scaling of $k_{CCZ}$ for large $m$.
Note that, as we will see shortly, in the regime where we ran the algorithm above ($m=9$)
the local lemma does {\it not} guarantee a solution.

Suppose that there are $n_{triple}=k_{CCZ}$ triples.
Imagine choosing each triple at random, following the initialization routine of the above algorithm.
Label the triples by an integer ranging $1,\dots,n_{triple}$.
Define a bad event $E_{i,j,k}$ to be the event that for three triples, labelled $i,j,k$, with $1 \leq i<j<k \leq n_{triple}$, there is a violation
of the triorthogonality conditions involving one indicator vector from each triple.
We call such events $E_{i,j,k}$ ``three-triple events".
Define a bad event $E_{i,j}$
 to be the event that for two triples, labelled $i,j$, with $1\leq i<j\leq n_{triple}$, there is a violation
of the triorthogonality conditions involving one indicator vector from one triple and two indicator vectors from the other triple.
We call such event $E_{i,j}$ ``two-triple events".

The probability of $E_{i,j,k}$ can be estimated as follows: There are $3^3=27$ different choices of indicator vectors if we choose one indicator vector from each triple.  The vector from the first triple is random.  The probability that the vector from the second triple has no overlap with the vector from the first triple is
\[
\binom{\frac{2}{3}m}{\frac{1}{3}m} \Big/ \binom{m}{\frac{1}{3}m}.
\]
Conditioned on the vectors from the first two triples having no overlap, the probability that the vector from the third triple has no overlap with either of the other two vectors is
\[
1\Big/\binom{m}{\frac{1}{3}m}.
\]
Thus,
\begin{align}
\Pr(E_{i,j,k})
~ \leq ~
27
 \frac{\binom{\frac{2}{3}m}{\frac{1}{3}m} }{ \binom{m}{\frac{1}{3}m}\binom{m}{\frac{1}{3}m} }
~ \simeq ~
2^{-m (2 H(1/3)-2/3)},
\end{align}
where $H(p)=-p\log_2(p)-(1-p)\log_2(1-p)$ is the binary entropy function 
and the approximate equality is up to subexponential factors.
Note $H(1/3)\approx 0.918$ and $2H(1/3)-2/3\approx 1.17$.

The probability of $E_{i,j}$ can be estimated as follows: 
There are $\binom{3}{1} \binom{3}{2} + \binom{3}{2}\binom{3}{1} = 18$ ways 
to choose one indicator vector from $i$ and two from $j$ or two from $i$ and one from $j$.
Suppose we choose two from $i$;
they have no overlap by construction and the probability that the vector from $j$
has no overlap with them is 
\[
1\Big/ \binom{m}{\frac{1}{3}m}.
\]
Thus,
\begin{align}
\Pr(E_{i,j}) 
~\leq~  
18
\frac{1}{\binom{m}{\frac{1}{3}m}} 
~\simeq~ 
2^{-mH(1/3)}.
\end{align}

We use the following statement of the Lovasz local lemma~\cite{alon2004probabilistic}.
Define a dependency graph on a set of events
such that two events are adjacent if and only if they are dependent.
For event $A$, let $\Gamma(A)$ denote the set of neighbors of $A$ in the dependency graph.
If one can choose a number $x(A)$ for each event $A$, $0\leq x(A)<1$, such that
for all $A$ we have
\begin{align}
\label{LLLeq}
\Pr(A) \leq x(A) \prod_{B \in \Gamma(A)} (1-x(B)),
\end{align}
then there is a nonzero probability that no event occurs.

Consider the dependency graph of all bad events (either three-triple or two-triple).
For bad events to be dependent, they must share one triple at least.
Hence, the neighborhood of any bad event (either three-triple or two-triple) 
includes $O(n_{triple}^2)$ three-triple events and $O(n_{triple})$ two-triple events.
Let us simply choose $x(A)=2 \Pr(A)$ for all bad events $A$.
Then, to have Eq.~\eqref{LLLeq}, it suffices to have that
$\prod_{B\in \Gamma(A)} (1-x(B)) \geq 1/2$,
which is implied by $\sum_{B\in \Gamma(A)} \pr(B) \leq 1/4$.
So, it suffices that
$n_{triple}^2 2^{-m(2 H(1/3)-2/3)}+n_{triple} 2^{-mH(1/3)}=O(1)$.
Thus, we want
\begin{align}
n_{triple} ~\lesssim ~ 2^{-m(H(1/3)-1/3)} ~\lesssim~ 2^{0.58 \ldots m}.
\end{align}

Therefore, if $k_{CCZ} = n_{triple} \lesssim 2^{0.58\ldots m}$ 
(neglecting subexponential factors in $m$),
the randomized algorithm finds a code $[[2^m,3k_{CCZ},2^{m/3}]]$
where the triorthogonality conditions for $CCZ$ are satisfied.

\subsection{Error Probabilities and Quantitative Values}
The generalized triorthogonal matrix has distance $d=2^{m/3}$.
The number of error patterns of weight $d$
which do not violate any stabilizer of the code is equal to 
the number of code words of $RM(2m/3,m)$ with weight $d$.
This is known\cite{sb} to equal
\begin{align}
A_d=\frac{2^m(2^m-2^0) (2^m-2^1)(2^m-2^2) \ldots (2^m-2^{\mu-1})}{2^\mu(2^\mu-2^0) (2^\mu-2^1)(2^\mu-2^2) \ldots (2^\mu-2^{\mu-1})},
\end{align}
where $\mu=m-r$ with in this case $r=2m/3$ so $\mu=m/3$.
For $m=3$, $A_d=28$.  
For $m=6$, $A_d=10416$.  
For $m=9$, $A_d=50434240 \approx 5 \times 10^7$.
The leading coefficient in the output error rate is of course at most these numbers,
since there could be $Z$-stabilizers of weight $d$.  Further, in the $m=6$ and $m=9$ cases above, we extended $G_0$
so the number of error patterns of weight $d$ is strictly smaller than $A_d$.
Indeed, for our maximal $m=6$ code, a direct enumeration shows
that there are 3248 error patterns that does not violate $X$-stabilizers,
out of which 304 are $Z$-stabilizers.

It is also known\cite{bs} that all weights of $RM(2m/3,m)$ between $d$ and $2d$ are of the form $2d-2^i$ for some $i$, 
so that the next weight after $d$ is equal to $3d/2$.  

To give some numbers when using these codes in a distillation protocol, 
consider the $m=9$ case with $k_{CCZ}=10$.
Suppose we have an input error probability $\ein=10^{-3}$.
Then, the probability that the protocol succeeds 
(i.e., that no stabilizer errors are detected) 
is lower bounded by $(1-\ein)^{512}\approx 0.599$.
The average number of output $CCZ$ magic states is then $\nCCZb \approx 5.99$.
We expect that for $m=9$ the contribution of errors with weight $3d/2=12$ 
will be negligible compared to the leading contribution.
Thus, we approximate that the output error probability by
$\eout \approx  A_d \ein^8 (1-\ein)^{504} \approx 3.0 \times 10^{-17}$.
where the factor $(1-\ein)^{504}$ represents 
the requirement that none of the other input $T$ gates have an error.
We expect that this is an overestimate because, as mentioned above, 
not all error patterns of weight $d$ that do not violate a stabilizer will lead to a logical error and also we have added additional stabilizers to $G_0$.
Thus, the ratio $\ebar = \eout / \nCCZb \approx 5.1 \times 10^{-18}$.
We use $512 / \nCCZb \approx 85.5$ $T$-gates per output $CCZ$ magic state.

It requires\cite{song2004optimal,jones2013low} $4$ high-quality $T$-gates to produce a single high-quality $CCZ$ state, so this
protocol's efficiency is comparable, if the goal is to produce $CCZ$ states,
to a protocol that uses only $85.5/4 \approx 21.4$ input $T$-gates per output $T$-gate.
Since one uses $4$ $T$-gates to make a $CCZ$ state, 
the quality of those output $T$-gates must be four times better than the needed $CCZ$ quality.

If one is able to improve the input error rate then the protocol becomes more efficient as the success probability becomes higher, asymptoting at $51.2$ $T$-gates per output $CCZ$ magic state, comparable to a protocol using $12.8$ input $T$-gates to produce an output $T$-gate.
Alternatively, one can also make the protocol more efficient by applying error correction as follows.  
Choose some integer $m\geq 0$.  
Then, modify the protocol; as usual, one encodes logical qubits in the $|+\rangle$ state into the error correcting code, 
applies a transversal $T$-gate, and then measures the stabilizers.  
However, while usually one would declare failure if any stabilizer errors occur, 
one can instead  apply error correction: 
if the error syndrome can be caused by at most $m$ errors, 
then one corrects those errors by applying Pauli $Z$ operators to the appropriate physical qubits.
For example, at $\ein=10^{-3}$, the probability that there are $0$ or $1$ input errors is equal to 
$(1-\ein)^{512}+512\ein(1-\ein)^{511} \approx 0.906$, 
giving the acceptance probability for $m=1$.  Applying
this error correction does reduce the quality of the output states: with $m=1$, 
now seven input errors can cause a logical error.
The number of such weight seven input error patterns that cause a logical error 
is at most $8 A_d$, so that the output error per output logical qubit is 
approximately $8 A_d \ein^7 / 10 \approx 5 \times 10^{-14}$.

\subsection{Punctured Reed-Muller Codes}
\label{prmc}

Motivated by the puncturing ideas of \ref{triopunc}, 
we have considered puncturing a Reed-Muller code.
Instead of using $RM(m/3-1,m)$ as before, we now consider $RM(r,3r+1)$.
This code is triorthogonal as before, 
and is maximal in the sense that $(G^{\wedge 2})^\perp = G_0$.
We then randomly puncture this code.
The codes we found numerically are listed in Tables~\ref{tb:prm1},\ref{tb:prm2}.
Observe that the coefficients $A_d$ in the output error probabilities 
are fairly small given the code lengths.

We found that there is a unique $d=5$ code that can be obtained by puncturing $RM(2,7)$;
it is $[[125,3,5]]$.
This was simple to check: Any three-puncture in $RM(r,m>1)$ is 
equivalent,\footnote{A punctured code from a Reed-Muller code
is determined by the isomorphism class under affine transformations 
of the set of points corresponding to the punctured coordinates in the $m$-dimensional unit hypercube,
since an affine transformation is an automorphism of $\FF_2[x_1,\ldots,x_m]/(x_1^2-x_1,\ldots,x_m^2-x_m)$.
Any three-point set in the unit hypercube is affinely independent,
and hence is affinely equivalent to any other three-point set.}
and we numerically verified that any four-puncture to $RM(2,7)$ gave $d = 4$.

Let us now explain our numerical techniques.

The number $k$ of logical qubits in each case in the tables was calculated after the puncture;
$k$ is equal to the number of punctures only if the submatrix of the generating matrix of $RM$ on the punctured coordinates is full rank.
The $Z$-distance, which is relevant to the distillation purposes,
is computed either by the MacWilliams identity applied to $X$-stabilizer weight enumerators
that are computed by brute force enumeration,
or by enumerating all $Z$-logical operators of a given weight.
The computed $Z$-distance is in fact the true code distance
since the $Z$-stabilizer group contains a subgroup associated with the bit strings of the $X$-stabilizer group.
The MacWilliams identity was an effective method 
especially when the base code was $RM(2,7)$ where there are only $29$ $X$-stabilizers prior to puncture.
For this base code, we simply did a random search, trying many different random punctures of the code, and selected good examples that we found.

When the base code was $RM(3,10)$, there are 176 $X$-stabilizers to begin with,
so the brute force enumeration of the $X$-stabilizer weight enumerator became prohibitive unless many coordinates were punctured.  Also, at larger distances ($\geq 5$), a guided search became more efficient than a random search among codes.
To solve both these problems, we used an ``unpuncturing'' strategy based on the following observation.
Let $G_0$ be a matrix whose rows represent $X$-stabilizers,
and suppose $G'$ is a matrix whose rows represent $X$-logical operators
such that any $Z$-logical operator of minimal weight $d$ anticommutes with at least one $X$-logical operator of $G'$.
Then, we consider a new $X$-stabilizer matrix $\begin{bmatrix} I & G' \\ 0 & G_0 \end{bmatrix}$.
We claim that this new code does not have any $Z$-logical operator of weight $\le d$.
The proof is simple: If the bit string $v$ of a $Z$-logical operator of weight $\le d$ 
have nonzero substring on the columns of $G_0$, then, by construction,
that substring must have weight at least $d$, but such a substring has odd overlap with some row of $G'$
which must be cancelled by the substring on the columns of $I$.
This forces the weight to be larger than $d$.
The construction of a new code by adding more stabilizers and qubits,
is precisely the inverse of the puncturing procedure (up to permutations of qubits),
hence the name ``unpuncturing.''

\begin{table}
\caption{Punctured Reed-Muller codes I.
In this table, the base code prior to puncturing is $RM(2,7) = [128,29,32]$.
The decimal integers are short-hand notation for the binary coordinate
that indexes bits in the Reed-Muller code;
e.g., ``3'' in the first example means 
that one has to puncture the bit labelled by $0000011 \in \FF_2^7$.
The number of $Z$-logical operators of weight $d$ 
is obtained by the MacWilliams identity
applied to the $X$-stabilizer weight enumerators.
Since the $Z$ stabilizer group in any case corresponds 
to a subspace of dual of the pre-puncture Reed-Muller code,
the minimal weight of any $Z$ stabilizer is at least 8.
Every $X$-stabilizer has weight a multiple of 8,
and there is a basis of $X$-logical operators such that each basis element has weight $7 \mod 8$.
Hence, the transversal $T$ becomes $T^\dagger$ on every logical qubit.
As a distillation protocol, the output error probability is $A_d p^d$
at the leading order where $p$ is the independent error probability of the input $T$ states.
}
\centering
\begin{tabular}{c}
\hline\hline
Code parameter $[[n,k,d]]$ and $A_d = \#$ ($Z$-logical operators of weight $d$) \\
Decimal representation of binary coordinates to puncture\\
\hline\hline
$[[114,14,3]]$, $A_3=30$, $n/k = 8.14$\\
3, 10, 19, 20, 64, 66, 72, 96, 99, 104, 110, 114, 115, 124\\
\hline
$[[112,16,3]]$, $A_3=96$, $n/k =7$ \\
6, 8, 13, 14, 17, 28, 29, 33, 44, 57, 65, 75, 79, 82, 106, 116\\
\hline
$[[109,19,3]],$ $A_3=324$, $n/k=5.73$\\
10, 15, 16, 17, 32, 39, 40, 41, 48, 59, 66, 69, 72, 81, 100, 102, 108, 120, 126\\
\hline
$[[118,10,4]]$, $A_4=210$, $n/k=11.8$\\
11, 17, 19, 59, 74, 76, 91, 99, 105, 110\\
\hline
$[[116,12,4]]$, $A_4 = 495$, $n/k = 9.67$\\
0, 31, 52, 61, 73, 94, 96, 112, 114, 115, 118, 120\\
\hline\hline
\end{tabular}
\label{tb:prm1}
\end{table}

\begin{table}
\caption{
Punctured Reed-Muller Codes II, continued from Table~\ref{tb:prm1}.
In this table, the base code prior to puncturing is $RM(3,10) = [1024,176,128]$.
The bound on $A_7$ of $[[937,87,7]]$ is from the exact weight enumerator (not shown) of $[[887,137,5]]$; we believe the true value of $A_7$ is much smaller based on the previous examples.
}
\centering
{\footnotesize
\begin{tabular}{c}
\hline\hline
$[[863,161,3]]$, $A_3 = 3231$, $n/k = 5.36$\\
3,4,7,10,15,39,42,44,45,49,59,66,68,70,72,74,91,103,104,109,119,120,122,123,130,161,164,170,\\
183,186,200,208,214,233,236,237,248,270,278,288,294,295,296,304,307,321,323,338,341,347,353,\\
356,359,360,365,374,377,404,411,414,425,443,447,455,465,470,474,477,480,482,492,493,502,507,\\
509,511,513,517,525,528,539,543,550,555,567,577,581,598,599,600,602,603,608,609,612,616,620,\\
621,628,638,646,652,659,660,669,678,681,687,714,728,738,739,741,743,744,745,748,750,758,768,\\
786,791,794,795,806,822,843,844,845,853,855,864,865,884,889,891,892,902,907,913,916,921,939,\\
942,
943,944,945,951,953,961,965,971,978,980,984,985,992,1002,1005,1012,1018\\
\hline
$[[872,152,4]]$, $A_4 = 1514$, $n/k = 5.74$\\
31,35,45,46,50,62,85,89,91,113,118,119,122,127,140,144,157,168,169,171,173,186,190,210,218,219,\\
228,230,237,244,249,254,263,271,281,282,308,336,352,353,398,404,405,411,412,441,444,455,456,460,\\
471,474,475,480,484,488,492,502,504,507,511,517,520,522,532,542,543,559,570,574,577,578,579,580,\\
583,592,598,601,602,605,608,612,615,618,620,637,643,644,653,658,667,688,690,694,714,717,724,727,\\
737,745,752,754,758,764,765,770,782,794,795,802,808,812,813,814,815,823,824,838,847,849,850,852,\\
861,863,867,871,874,880,901,907,911,915,919,921,924,926,941,950,954,969,971,972,976,977,982,991,\\
995,999,1008,1013,1014,1023\\
\hline
$[[887,137,5]]$, $A_5 = 709$, $n/k = 6.47$\\
11,21,30,37,39,53,68,74,78,82,98,105,107,120,130,136,148,149,152,161,162,163,181,194,209,210,211,\\
233,234,243,244,267,269,274,277,281,284,298,317,324,325,329,341,361,362,375,389,399,400,405,412,415,\\
423,425,449,480,487,495,507,511,522,538,542,557,563,578,579,584,593,600,609,610,619,622,623,635,638,\\
639,640,643,644,651,653,655,657,661,671,672,678,680,692,714,727,737,775,777,792,796,806,817,826,827,\\
831,833,834,837,851,852,854,857,866,868,871,875,880,890,891,896,897,898,916,924,936,938,941,958,964,\\
965,966,973,975,983,984,990,996,997,1022\\
\hline
$[[912,112,6]]$, $A_6 = 1191$, $n/k = 8.14$\\
11,21,37,39,68,74,78,82,98,107,130,148,152,161,162,163,181,194,209,210,211,233,243,244,267,269,274,\\
277,298,317,324,325,329,341,361,362,399,405,412,415,423,425,480,487,495,507,522,542,557,563,579,584,\\
593,600,609,610,619,622,623,635,639,640,653,655,657,661,671,672,678,680,692,714,727,737,775,777,792,\\
796,806,826,827,831,833,834,837,851,852,854,857,866,871,875,880,890,891,896,897,898,916,924,936,938,\\
941,958,965,966,983,984,990,996,997,1022\\
\hline
$[[937,87,7]]$, $A_7 \le 1887905$, $n/k = 10.77$\\
21,37,39,68,74,82,98,130,148,152,162,163,194,209,210,211,233,244,267,269,274,317,324,325,329,341,361,\\
362,399,405,412,415,423,480,487,495,507,522,542,557,563,584,593,600,609,610,623,635,639,640,657,661,671,\\
672,692,714,727,737,777,792,796,826,827,831,833,834,837,851,852,854,857,871,875,880,890,891,896,897,898,\\
924,936,958,966,984,996,997,1022\\
\hline\hline
\end{tabular}
}
\label{tb:prm2}
\end{table}

For small distances, e.g., $d = 3$, it is easy to enumerate all $Z$-logical operators of weight $d$.
We then select $X$-logical operators to ``catch'' those minimal weight $Z$-logical operators,
and identify the punctured coordinates that gave rise to the chosen $X$-logical operators.
One $X$-logical operator $\bar X$ was chosen each time so that
the number of the mimimal weight $Z$-logical operators
that $\bar X$ anticommutes with is maximized.
The codes in Table~\ref{tb:prm2} were found by this unpuncturing.
We started with a random puncturing giving a $d=3$ code and then successively unpunctured to obtain distance $4,5$ codes.
The $d=6$ and $d=7$ codes in Table~\ref{tb:prm2} were obtained by unpuncturing the best rate code with $d=5$ that we found.
Note that for the code $[[937,87,7]]$,
it was prohibitively costly to enumerate all logical operators of weight $7$,
so we contented ourselves by an upper bound on the number of $Z$-logical operators.
The bound was possible since we computed, by brute force, the $X$-stabilizer's weight enumerator of $[[887,137,5]]$,
unpuncturing which yielded $[[937,87,7]]$; 
while in general this $X$-stabilizer weight enumerator is very costly to compute as we explained above, 
it was possible to compute it for a single code example 
(it would not be practical to compute this enumerator for all the codes tried in a random search).

To give some numbers when these codes are used in a distillation protocol, 
consider a $\ein=10^{-3}$ input error rate using the $[[912,112,6]]$ code.
In this case, the probability that the protocol succeeds is at least 
$p_{acc} = (1-\ein)^{912}=0.401$.
The average number of output $T$ magic states is then $\nTb \approx 44.97$.
We expect that the dominant contribution to the errors is from the leading order so
we approximate the output error probability per output state by
$\ebar \approx  A_6 \ein^6 (1-\ein)^{906} / \nTb \approx 1.07 \times 10^{-17}$.
We use $912 / \nTb \approx 20.28$ $T$-gates per output $CCZ$ magic state.
One can also use error correction to increase the success probability
at the cost of an increase in output error rate.
If one corrects a single error, the acceptance probability becomes approximately
$p_{acc} = (1-\ein)^{912}+912\ein(1-\ein)^{911} \approx 0.768$.  
Applying this error correction does reduce the quality of the output states 
since now five input errors can cause a logical error.
The number of such weight five input error patterns that cause a logical error 
is at most $6 A_6$, so that the output error per output logical qubit is 
approximately $6 A_6 \ein^5(1-\ein)^{907} / 112 p_{acc} \approx 3.35 \times 10^{-14}$ 
with a number of input states per output state of approximately $10.60$.
These ratios of input to output states are better 
than any protocol we know with the given input and output error rates; 
further, the performance of the punctured codes will improve 
at lower input errors where the success probability becomes closer to $1$.
We expect that for the $[[937,87,7]]$ code one can find even lower output error rates.

We have explained a distillation protocol that is particularly well-suited 
for any punctured (and hence for all) triorthogonal code
in Section~\ref{triopunc}~\cite{Fowler2012}.
Since
$RM(r,3r+1)$ is triply even, the Clifford correction after applying transversal $T$ is absent and so
the only Clifford cost for the present punctured Reed-Muller codes 
is in the preparation of the stabilizer state
$\sum_{v \in RM(r,3r+1)} \ket{v}$.

The Clifford circuit to prepare this stabilizer state is a coherent version of a classical encoding circuit,
and there exists an encoding circuit of depth $m$ using $(m/2) 2^{m}$ CNOTs
(if one can implement CNOT across any pair of qubits),
using the recursive construction of Reed-Muller codes.  
This circuit can be described as follows: using $2^m$ qubits labelled by bit strings of length $m$, prepare all qubits labelled by bit strings with Hamming weight $\leq r$ in the $|+\rangle$ state and prepare all other qubits in the $|0\rangle$ state.  Then, for $m$ rounds, labelled by integers $1,\ldots,m$, do the following: on the $j$-th round, for each of the $2^{m-1}$ qubits labelled by a bit string with a $0$ in the $j$-th position of that bit string, apply a CNOT with that qubit as source and with the target being the qubit labelled by the bit string which agrees everywhere with the source bit string, except that it is $1$ in the $j$-th position.
This circuit is the same as the encoding circuit used for polar codes, up to different choices of the input state~\cite{arikan2009channel,renes2012efficient}.

\section{$T$-to-$CCZ$ protocols using hyperbolic weakly self-dual CSS codes}
\label{sec:hyper}

In Ref.~\cite{HHPW},
we have classified weakly self-dual CSS codes on $\nin$ qubits into two types.
If $\cS$ is the self-orthogonal subspace of $\FF_2^n$ 
corresponding to the stabilizers of the code,
the distinction criterion is whether $\cS$ contains all-1 vector $\vec 1$.
If $\vec 1 \in \cS$,
the space of representing logical operators $\cS^\perp / \cS$ is hyperbolic,
and the parameters $\nin$, $\kin$, and the code distance must be even numbers.
For hyperbolic codes, the binary vector space corresponding to the logical operators
is isomorphic to direct sum of hyperbolic planes.
Here, we only consider hyperbolic codes.
Choose a basis $\{ \ell^{(1)}, \ell^{(2)}, \ldots, \ell^{(\kin)} \}$ 
of $\cS^\perp/ \cS$ such that the dot product between the basis vectors satisfy%
\footnote{
A weakly self-dual CSS code, defined by a self-orthogonal binary subspace $\cS$,
is hyperbolic if and only if $\cS^\perp / \cS$ admits such a basis.
See \cite[Sec.~3]{HHPW} for further details.
}
\begin{align*}
    \ell^{(2a-1)} \cdot \ell^{(2b-1)} &= 0\\
    \ell^{(2a)} \cdot \ell^{(2b)} &= 0 & \text{for } a ,b=1,\ldots,\kin/2 \\ 
    \ell^{(2a-1)} \cdot \ell^{(2b)} &= \begin{cases} 1 & \text{if } a = b ,\\ 0 & \text{otherwise.} 
    \end{cases}\label{eq:hyperbolicpairing}
\end{align*}
We call such a basis hyperbolic,
and Gram-Schmidt procedure can be used to find a hyperbolic basis.
We define logical operators as
\begin{align}
    \tilde X_{2a-1} &= X(\ell^{(2a-1)}),& \tilde Z_{2a-1} &= Z(\ell^{(2a)}), \\
    \tilde X_{2a} &= X(\ell^{(2a)}), & \tilde Z_{2a} &= Z(\ell^{(2a-1)}), & 
    \text{for } a = 1,\ldots,\kin/2. \nonumber
\end{align}
Note that this is different from the magic basis of Ref.~\cite{HHPW}
where a pair of logical qubits are swapped under the transversal Hadamard.

We now investigate the action of transversal $\sgate$ gate.
Since $\sgate X \sgate^\dagger = Y = -iZX$, unless $\nin$ is a multiple of 4,
the transversal $\sgate$ is not logical.
However, there is a simple way to get around this.
Instead of applying $\sgate$ on every qubit,
we assign exponents $t_i = \pm 1$ to each qubit $i$,
which depends on the code,
and apply $\bigotimes_i \sgate^{t_i}$.
We choose $t_i$ such that
\begin{align}
    \sum_i v_i t_i &= 0 \mod 4 \text{ for any } v \in \{ b^{(1)}, \ldots, b^{(\dim \cS)} \} \subset \cS,\\
    \sum_i \ell^{(a)}_i t_i &= 0 \mod 4 \text{ for } a = 1,\ldots,\kin,
\end{align}
where it is implicit that the elements of $\FF_2$ are promoted to usual integers
by the rule that $\FF_2 \ni 0 \mapsto 0 \in \mathbb Z$ and $\FF_2 \ni 1 \mapsto 1 \in \mathbb Z$,
and $\{ b^{(j)} \}$ is a basis of the $\FF_2$-vector space $\cS$.
A solution $t_i$ to these conditions always exists,
because the Gauss elimination for the system of equations over $\mathbb Z / 4\mathbb Z$,
never encounters division by an even number when applied to a full $\FF_2$-rank matrix.
Once we have a valid $t_i$, then it follows that $\sum_i v_i t_i = 0 \mod 4$ 
for {\it any} vector $v \in \cS$.
Since any vector is a sum of basis vectors,
which are orthogonal with one another,
this follows from the following identity.
For any integer vector $y$~\cite{Ward1990,bh}
\begin{align}
    \sum_i y_i \mod 2 &= \sum_i y_i - 2 \sum_{i < j} y_i y_j &\mod 4 \label{eq:mod4},\\
    \sum_i y_i \mod 2 &= \sum_i y_i - 2 \sum_{i < j} y_i y_j + 4 \sum_{i<j<k} y_i y_j y_k &\mod 8. \label{eq:mod8}
\end{align}
Likewise, for {\it any} vector $\ell \in \cS^\perp$ and any $s \in \cS$,
we have $\sum_i \ell_i t_i = \sum_i (\ell + s \mod 2)_i t_i \mod 4$.

We now show that the action of $\bigotimes_i \sgate^{t_i}$ on 
the logical state $\ket{\tilde x_1,\ldots,\tilde x_\kin}$ is
control-$Z$ on hyperbolic pairs of logical qubits:
\begin{align}
&\left(\bigotimes_i \sgate^{t_i} \right) \ket{\tilde x_1, \ldots, \tilde x_\kin} \nonumber\\
&= \frac{1}{\sqrt{|\cS|}} \sum_{s \in \cS} e^{(i\pi/2) \sum_j f_j t_j} 
\ket{f = s + x_1 \ell^{(1)} + \cdots + x_\kin \ell^{(\kin)} \mod 2 }\nonumber\\
&= e^{(i\pi/2) 
\left( \sum_{a} x_a \sum_j \ell^{(a)}_j t_j - 2 \sum_{a<b} x_a x_b \sum_j \ell^{(a)}_j\ell^{(b)}_j t_j \right)
}\times \nonumber\\
& \qquad \qquad
\frac{1}{\sqrt{|\cS|}} \sum_{s \in \cS}
\ket{s+ x_1 \ell^{(1)} + \cdots + x_\kin \ell^{(\kin)} \mod 2 }\nonumber\\
&= \left( \prod_{j=1}^{\kin/2} \widetilde{CZ}_{2j-1,2j} \right) 
  \ket{\tilde x_1, \ldots, \tilde x_k}\label{eq:Saction}
\end{align}
where 
in the third line we used \eqref{eq:mod4} 
and in the last line we used \eqref{eq:hyperbolicpairing}.

Therefore,
if we implement control-$\sgate$ gate over a hyperbolic code,
then we implement a measurement routine for product of $CZ$ operators.
The control-$S$ can be implemented using an identity
\begin{align}
    ^C \sgate &= (^C e^{i\pi/4}) T(^C X)T^\dagger (^C X)
\label{eq:TXTX}
\end{align}
where $^C U = \ket 0 \bra 0 \otimes I + \ket 1 \bra 1 \otimes U$;
in particular, $^C e^{i \pi /4} = \left( \ket 0 \bra 0 + e^{i \pi / 4} \ket 1 \bra 1 \right) \otimes I$.
Since a hyperbolic CSS code contains $\onev$ in the stabilizer group,
we know $\sum_i t_i \vec 1_i = 0 \mod 4$, and
the control-phase factor will either cancel out or become $Z$ on the control.
If $T$ gates in this measurement routine are noisy with independent $Z$ errors of probability $p$,
then upon no violation of stabilizers of the hyperbolic code,
the measurement routine puts $O(p^2)$ error into the measurement ancilla,
and $O(p^d)$ error into the state under the measurement
where $d$ is the code distance of the hyperbolic code.

\subsection{Quadratic error reduction}

The control-$Z$ action on the logical level can be used to implement control-control-$Z$,
whenever the hyperbolic code is encoding one pair of logical qubits.
The smallest hyperbolic code that encodes one pair of logical qubits 
is the 4-qubit code of code distance $2$, with stabilizers $XXXX$ and $ZZZZ$.
The choice of logical operators that conforms with our hyperbolic conditions is
\begin{align}
    \begin{pmatrix}
        X & X & I & I \\
        I & Z & Z & I \\
        \hline
        I & X & X & I \\
        Z & Z & I & I 
    \end{pmatrix}.
\end{align}
The exponents $t_i$ for $\sgate$ is thus
\begin{align}
t = \begin{pmatrix} + & - & + & - \end{pmatrix}.
\end{align}
Using this choice of $t_i$, the phase factor in \eqref{eq:TXTX} cancels out.

Every non-Clifford gates enters the circuit by \eqref{eq:TXTX},
and hence any single error will be detected.
Since the ancilla that controls $\sgate$ inside the hyperbolic code can be contaminated
by a pair of $T$ gates acting on the same qubit,
there is little reason to consider hyperbolic code of code distance higher than 2.
When applied to $\ket{+^{\otimes 3}}$,
the routine described here outputs one $CCZ$ state using 8 $T$-gates,
with output error probability $28p^2 + O(p^3)$ 
where $p$ is the independent error rate of $T$ gates.

The overall circuit is very similar to quadratic protocol in
Jones~\cite{Jones2013composite}
in which the same choice of logical operators are used,
but control-$(TXT^\dagger)^{\otimes 4}$ is applied on the code,
followed by syndrome measurement and then 
$\pi/2$ rotation along $x$-axis on the Bloch sphere.
In contrast, 
we apply control-$(TXT^\dagger X)_1(XTXT^\dagger)_2(TXT^\dagger X)_3(XTXT^\dagger)_4$,
and then syndrome measurement, without any further Clifford correction.

\subsection{Quartic error reduction}

For a higher order error suppression of $CCZ$ states,
we use the hyperbolic codes to check the eigenvalue of
the stabilizers of the $CCZ$ state $\ket{CCZ} = CCZ\ket{+^{\otimes 3}}$.
The stabilizers are $(CZ)_{12}X_3$, $(CZ)_{13}X_2$, and $(CZ)_{23}X_1$.
(These are obtained by conjugating $X_{1,2,3}$, 
the stabilizers of $\ket{+^{\otimes 3}}$, by $CCZ$ gate.)
As there are three stabilizers, we need three rounds of checks.
By symmetry, it suffices to explain how to measure $(CZ)_{12}X_3$.

Suppose we have a hyperbolic weakly self-dual CSS code of parameters $[[\nin,2k,4]]$.
This is our inner code~\cite{HHPW}.
(For example, there is a quantum Reed-Muller code 
of parameters $[[2^m, 2^m-2m-2, 4]]$ for any $m \ge 4$.
There are also Majorana codes which can be interpreted 
as hyperbolic codes on qubits~\cite{BravyiLeemhuisTerhal2010Majorana,Hastings,RMqubit}.)
Take $k$ independent output $CCZ$ states from the quadratic protocol in the previous subsection,
and separate a single qubit from each of the $CCZ$ states.
On these separated qubits we act by $^C X$ with a common control.
The remaining $2k$ qubits are then embedded into the hyperbolic code,
with which $^C (CZ)$ will be applied on the logical qubits,
using $2n$ $T$ gates with independent error probability $p$.
It is important that the control qubit is common for all controlled gates.
This way, the product of $k$ stabilizers on the $k$ $CCZ$ states are measured.

One has to run this check routine three times
for each of the three stabilizers of $CCZ$ states.
In total, the number of input $T$ gates is $8k + 6\nin$ 
where $8k$ is from the protocol in the previous subsection,
and $3 \cdot 2\nin$ is inside the distance-4 hyperbolic inner code.

Upon no stabilizer violations of the inner code and outer code measurements,
the protocol outputs $k$ $CCZ$-states.
If the inner hyperbolic code does not have any error on $T$ gates while implementing $^C(CZ)$,
then the output $CCZ$ states' error rate is quadratic in the input $CCZ$ states' error rate.
This being quadratic is due to the fact that we have an outer code of code distance 2.
(An outer code is one that specifies which input states to check.
See \cite{HHPW} for detail.)
Thus, the output error from this contribution is $\binom{k}{2} (28p^2)^2$ at the leading order.

There could be a pair of errors in the $T$ gate inside the inner code
that flips the eigenvalue measurement of $(CZ)X$.
In order for this type of error to be output error there must be an odd number of
errors in the input $CCZ$ states.
Hence, the contribution to the output error probability is 
$k \cdot 28p^2 \cdot 3n p^2$ at leading order.

Finally, the inner code may have 4 errors leading to logical errors
since the code distance is 4.
An upper bound on this contribution to the output error probability is
$3 \cdot 2^3 A_4p^4$, 
where $A_4$ is the number of $Z$ logical operators of the inner code of weight 4.
The factor of $2^3$ is because one $Z$ error on a qubit of the inner code can occur
in one of two places, and the half of all such configurations lead to an accepted output.
This is likely an overestimate because a logical error from a check out of three checks
can be detected by a later check.
In case of the Reed-Muller codes, we see
$A_4([[16,6,4]]) = 140$,
$A_4([[32,20,4]]) = 620$, and
$A_4([[64,50,4]]) = 2604$.

Using $[[16,6,4]]$,
the output error probability has leading term at most $9744 p^4$ or
$\ebar = 3.2 \times 10^3 p^4$ per output,
and the input $T$ count is $\nTbar = 40$ per output $CCZ$.
This particular protocol is worse in terms of input $T$ count 
than the protocol by a generalized triorthogonal code above, the protocol of \cite{Jones2013composite},
or a composition of quadratic protocols of $T$-to-$T$~\cite{bh} and 8$T$-to-1$CCZ$~\cite{campbell2017unified},
but better in terms of space footprint ($< 25$ qubits).
Using $[[32,20,4]]$ we see $\ebar \approx (7.7 \times 10^3) p^4$ and $\nTbar = 27.2$.
Using Reed-Muller $[[64,50,4]]$, we see $\ebar = (4.3 \times 10^4) p^4$ and $\nTbar = 23.4$.
For large $m$, i.e., encoding rate near 1,
the input $T$ count approaches 20 per output $CCZ$.

We have ignored the acceptance probability.
Since the input $CCZ$ states can be prepared independently using only 8 $T$ gates,
we may assume that the preparation is always successful.
Termination of the protocol is due to nontrivial syndrome on the distance 4 code.
Since there are $6n$ $T$ gates, 
the overall acceptance probability is at least
$(1-p)^{6n}$.

In the next section, we present another family that has even lower asymptotic input $T$ count.

\section{Clifford stabilizer measurements using normal weakly self-dual CSS codes}
\label{sec:normal}

As an extension of Ref.~\cite{HHPW},
we can turn $m$-copies of any normal weakly self-dual  CSS code (normal code)
into a measurement routine of magic states of form $U\ket{+^{\otimes m}}$
where $U$ belongs to the third level of Clifford hierarchy.
This is based on the observations that such states have Clifford stabilizers of form $V_i = UX_i U^\dagger$,
which can be measured by controling the middle $X_i$,
and that any normal code admits transversal implementation of logical $V_i$.
For the clarity of presentation,
we will explain protocols for distilling $CCZ$ states,
and leave general cases to the readers.
If $V_i$ involves $S$-gates, one has to choose appropriate exponents $t_i = \pm 1$ 
such that $\otimes_i S_i^{t_i}$ on physical qubits becomes a logical $S$-gate;
see the previous section.

Recall that a normal code is a weakly self-dual CSS code,
defined by a self-orthogonal binary vector space $\cS$
such that $\onev \notin \cS$.
In such a code the binary vector space $\cS/\cS^\perp$
corresponding to the logical operators,
has a basis such that any two distinct basis vectors have even overlap (orthogonal)
but each of the basis vector has odd weight.
Associating each basis vector to a pair of $X$- and $Z$-logical operators,
we obtain a code where the transversal Hadamard 
induces the product of all logical Hadamards.

Observe that in a normal code the transversal $X$ anti-commutes with
every $Z$ logical operator, and hence is equal to, up to a phase factor,
the product of all $X$ logical operator.
In the standard sign choice of logical operators 
where every logical $X$ is the tensor product of Pauli $X$,
the transversal $X$ is indeed equal to the product of all $X$ logical operators.
Likewise, the transversal $Z$ is equal to the product of all $Z$ logical operators.
Then, it follows that control-$Z$ across a pair of identical normal code blocks 
is equal to the product of control-$Z$ operators over the pairs of logical qubits.

Therefore, given three copies, labeled $A,B,C$, of a normal code $[[\nin,\kin,d]]$,
if we apply $\bigotimes_{i=1}^{\nin} CZ_{Ai,Bi}X_{Ci}$,
then the action on the code space is equal to
$\bigotimes_{j=1}^{\kin} \widetilde{CZ}_{Aj,Bj} \tilde X_{Cj}$.

Having a transversal operator that induces 
the action of the stabilizer $(CZ)X$ of $CCZ$-state on the logical qubits,
we will make a controlled version of this.
We use the following identity:
\begin{align}
    (CCZ)_{123} (^{C_a}X_1) (^{C_b}X_2) (^{C_c}X_3) (CCZ)_{123} 
    = 
\left[ ^{C_a}(CZ_{23}X_1) \right]
\left[ ^{C_b}(CZ_{13}X_2) \right]
\left[ ^{C_c}(CZ_{12}X_3) \right]
\label{eq:3stabCCZ}
\end{align}
which is the product of three stabilizers of $CCZ$-state
controlled by three independent ancillas.
The transversality of the logical operator $(CZ)X$ implies that
if we apply \eqref{eq:3stabCCZ} transversally across a triple of normal codes,
then the three ancillas will know the eigenvalue of the three stabilizers of $CCZ$, respectively.
The non-Clifford gate $CCZ$ in \eqref{eq:3stabCCZ} can be injected
using $4$ $T$-gates~\cite{song2004optimal,jones2013low}.

This method of measuring stabilizers of $CCZ$ state,
compared to that in the previous section using the hyperbolic codes,
has advantage that one does not have to repeat three times for each of three stabilizers,
but has disadvantage that one needs roughly a factor of three space overhead.
(The space overhead comparison is not completely fair, 
because a code cannot be simultaneously normal and hyperbolic.
However, in the large code length limit this factor of 3 
in the space overhead is appropriate.)
In the large code length limit, this method also has an advantage in terms of $T$-count.
Using the hyperbolic codes, even if the encoding rate is near one,
we need 12 $T$ gates per $CCZ$-state under the test.
On the other hand, using \eqref{eq:3stabCCZ} on a normal code of encoding rate near one,
we need 8 $T$ gates per $CCZ$-state under the test.

Now the protocol at quartic order is as follows.
Prepare $\kin$ $CCZ$-states from the quadratic protocol using 4-qubit code.
This consumes $8\kin$ $T$-gates with independent error probability $p$.
Embed them into the triple of normal code of parameter $[[\nin,\kin,4]]$
with each qubit of the $CCZ$ states into a different code block.
Apply \eqref{eq:3stabCCZ}; this step consumes $8\nin$ $T$ gates 
with independent error probability $p$.
Upon no violation of code's stabilizers and ancillas,
decode the logical qubits and output $\kin$ $CCZ$ states.

This is a quartic protocol
as the output is faulty only if
(i) an even number of input $CCZ$ states are faulty, which happens at order $(p^2)^2$,
(ii) an odd number of input $CCZ$ states are faulty but missed by a flipped ancilla outcome,
which happens at order $p^2 \cdot p^2$,
(iii) some error in the inner code is a logical error,
which happens at order $p^d = p^4$,
or (iv) some other error of higher order occurred.
The total number of $T$ gates used is $8\kin + 8 \nin$.

There are normal codes of encoding rate greater than $2/3$ 
and code distance 4 or higher on tens of qubits,
including quantum BCH codes $[[63,45,4]]$~\cite{qbch} 
and ``$H$-codes'' of parameters $[[k^2+4k+4,k^2,4]]$ where $k$ is even~\cite{Jones2012}.
Random constructions~\cite{CalderbankShor1996Good,HHPW} 
guarantee such codes of encoding rate near one in the limit of large code length.
The input $T$ count in the current quartic protocol 
using a high rate inner code approaches $16$ per output $CCZ$.

In terms of input $T$ count, to the best of the authors' knowledge,
this family is better than any previous $T$-to-$CCZ$ protocol
with quartic order of error reduction.

We can bootstrap the protocol to have a family of protocols 
for $d=2^\alpha, \alpha \ge 2$.
The construction is inductive in $\alpha$.
Fix an inner code $[[\nin,\kin,4]]$. 
(This is for simplicity of presentation, and is not necessity.)
The quartic protocol above is the base case in the induction.
Suppose we have constructed a $2^\alpha$-th order protocol $P_\alpha$
and a $2^{\alpha-1}$-th order protocol $P_{\alpha-1}$
using $n_\alpha$, $n_{\alpha-1}$ $T$ gates per output $CCZ$,
respectively.
The protocol is then:
(1) Run $P_\alpha$ many times to prepare independent input states at error rate $p^{2^{\alpha}}$.
(2) Embed them into the triples of the inner code.
(3) Apply \eqref{eq:3stabCCZ} 
where $CCZ$-gates are injected by outputs from $P_{\alpha-1}$.
(4) Upon no violation of the code's stabilizers, output the logical qubits.
The order of reduction in error can be seen by considering the cases 
(i), (ii), and (iii) above.
In all cases, the order of the error is $2 \cdot 2^{\alpha} = 2^{\alpha+1}$,
$2^{\alpha} \cdot (2 \cdot 2^{\alpha-1}) = 2^{\alpha+1}$,
or $4 \cdot 2^{\alpha-1} = 2^{\alpha+1}$.
Step (1) takes $n_\alpha$ $T$-gates per $CCZ$ state by induction hypothesis.
For $\kin$ sufficiently close to $\nin$,
step (3) takes $2 n_{\alpha-1}$ $T$-gates per $CCZ$. Hence,
$n_{\alpha+1} \simeq n_\alpha + 2 n_{\alpha-1}$, and
\begin{align}
    n_\alpha \simeq 4 \cdot 2^{\alpha} = 4 d 
\end{align}
since $n_1 = 8$ and $n_2 = 16$.

It is possible to combine the idea presented here 
with that of \cite{campbell2017unified} to reduce the input $T$-count
at the expense of dealing with a larger batch.
At $d=2$, \cite{campbell2017unified} has asymptotic $T$-count $n_1 = 6$.
At $d=4$, using a high encoding rate normal code,
the input $T$-count approaches $n_2 = 8+6=14$, instead of $16$.
At $d=8$, the count becomes $14+2 \cdot 6 = 26$, instead of $32$.
At a larger $d$ that is a power of 2, 
in the limit of large code length,
the $T$-count approaches $(2/3)(5 \cdot 2^\alpha + (-1)^\alpha)$
which is at most $3.33 d + 0.67$.

Note that this bootstrapping for large $\alpha$
must involve a quite large number of qubits
to ensure the independence of the input $CCZ$ states,
and the $CCZ$-gates on the inner code.
The usage of \cite{campbell2017unified} further enlarges the necessary batch size.

We finally note that 
for $d = 2d' \ge 10$ with $d'$ odd,
one can first use the protocol of \cite{HHPW} 
to produce a $T$ gate with error at $d'$-th order
where $d'\ge 5$ is odd using $d' + o(1)$ $T$ gates per output $T$, and then 
use $T$-to-$CCZ$ protocol of \cite{campbell2017unified}
to have $CCZ$ states with error at $(d=2d')$-th order.
This combination will give $T$ count $3d$ per output $CCZ$.

\bibliographystyle{apsrev4-1}
\nocite{apsrev41Control}
\bibliography{rtrio-ref}

\end{document}